\documentclass[12pt]{article}
\newlength{\mytopmargin}
\newlength{\myleftmargin}
\usepackage{amsmath,amsfonts,amssymb,amsthm,graphicx,xcolor,physics}

\theoremstyle{plain}
\newtheorem{theorem}{Theorem}

\newtheorem{prop}[theorem]{Proposition}

\theoremstyle{definition}

\theoremstyle{remark}
\newtheorem{remark}[theorem]{Remark}

\numberwithin{equation}{section}

\begin{document}

\begin{center}
{\bfseries \LARGE Computation of marginal eigenvalue distributions in the Laguerre and Jacobi
$\beta$ ensembles}\\[2\baselineskip]
{\Large Peter J. Forrester${}^1$ and Santosh Kumar${}^2$}\\[.5\baselineskip]

\end{center}

\noindent
${}^1$School of Mathematics and Statistics, The University of Melbourne, Victoria 3010, Australia.\: \: Email: {\tt pjforr@unimelb.edu.au}; 
\\
${}^2${
Department of Physics, Shiv Nadar Institution of Eminence, Gautam Buddha Nagar, Uttar Pradesh - 201314,
India.\: \: Email: {\tt skumar.physics@gmail.com}

\begin{abstract}
We consider the problem of the exact computation of the marginal eigenvalue distributions in the Laguerre and Jacobi $\beta$ ensembles. In the case $\beta=1$ this is a question of long standing in the mathematical statistics literature. A recursive procedure to accomplish this task is given for $\beta$ a positive integer, and the parameter $\lambda_1$ a non- negative integer. This case is special due to a finite basis of elementary functions, with coefficients which are polynomials. In the Laguerre case with $\beta = 1$ and $\lambda_1 + 1/2$ a non-negative integer some evidence is given of their again being a finite basis, now consisting of elementary functions and the error function multiplied by elementary functions. Moreover, from this the corresponding distributions in the fixed trace case permit a finite basis of power functions, as also for $\lambda_1$ a non-negative integer. The fixed trace case in this setting is relevant to quantum information theory and quantum transport problem, allowing particularly the exact determination of Landauer conductance distributions in a previously intractable parameter regime. Our findings also aid in analyzing zeros of the generating function for specific gap probabilities, supporting the validity of an associated large $N$ local central limit theorem.
\end{abstract}

\section{Introduction}
Familiar in random matrix theory are the eigenvalue probability density functions (PDFs) for the Laguerre and Jacobi $\beta$-ensembles
\begin{equation}\label{1.1}
    \mathcal P^{\rm L}(x_1,\dots,x_N) = {1 \over Z_N^{\rm L}}
    \prod_{l=1}^N x_l^{\lambda_1} e^{-\beta x_l/2} \chi_{x_l > 0}
    \prod_{1 \le j < k \le N} |x_k - x_j|^\beta
\end{equation}
and
\begin{equation}\label{1.2}
    \mathcal P^{\rm J}(x_1,\dots,x_N) = {1 \over Z_N^{\rm J}}
    \prod_{l=1}^N x_l^{\lambda_1} (1 - x_l)^{\lambda_2}  \chi_{0< x_l < 1}
    \prod_{1 \le j < k \le N} |x_k - x_j|^\beta.
\end{equation}
Here the notation $\chi_A$ is for the indicator function, with $\chi_A=1$ for $A$ true and $\chi_A=0$ otherwise. For the normalisations $Z_N^{\rm L}, Z_N^{\rm J}$ to be well defined we require $\lambda_1 > -1$ in (\ref{1.1}) and $\lambda_1, \lambda_2 > -1$ in (\ref{1.2}), and further we take $\beta > 0$. One has then~\cite{Fo10},
\begin{align}\label{1.3a}
Z_N^{\rm L}=\left(\frac{\beta}{2}\right)^{-N(\beta(N-1)/2+\lambda_1+1)}\prod_{j=0}^{N-1} \frac{\Gamma(\beta j/2+\lambda_1+1)\Gamma(\beta(j+1)/2+1)}{\Gamma(\beta/2+1)},
\end{align}
\begin{align}\label{1.3b}
Z_N^{\rm J}=\prod_{j=0}^{N-1}\frac{\Gamma(\beta j/2+\lambda_1+1)\Gamma(\beta j/2+\lambda_2+1)\Gamma\left(\beta (j+1)/2+1\right)}{\Gamma\left(\beta(j+N-1)/2+\lambda_1+\lambda_2+2\right)\Gamma(\beta/2+1)}.
\end{align}

Random matrices realising (\ref{1.1}) and (\ref{1.2}) as their eigenvalue PDFs in the full parameter range were first given in 
\cite{DE02} and \cite{ES08} respectively. Alternative recursive constructions were given in \cite{FR02b}. From classical random matrix theory we know that with $\beta = 1,2$ and 4 and particular $\lambda_1, \lambda_2$ there are matrices with Gaussian entries with eigenvalue PDFs (\ref{1.1}) and (\ref{1.2}). Specifically, for $\beta = 1$ or $2$ define the $n_1 \times N$ ($n_1 \ge N$) rectangular random matrix $X_{n_1,N}^{(\beta)}$, with independent standard real ($\beta = 1$) and complex ($\beta = 2$) normal entries. Construct from this the positive definite random matrix $W_{n_1,N}^{(\beta)} = (X_{n_1,N}^{(\beta)})^\dagger X_{n_1,N}^{(\beta)}$, known as a real or complex Wishart matrix. One has that the eigenvalue PDF of $W_{n_1,N}^{(\beta)}$ is given by (\ref{1.1}) with $\lambda_1 = (\beta/2) (n_1 - N + 1) - 1$; see e.g.~\cite[Prop.~3.2.2]{Fo10}. Moreover, using the Wishart matrices to now form $Y_{n_1,n_2,N}^{(\beta)} = (\mathbb I_N + (W_{n_1,N}^{(\beta)})^{-1} W_{n_2,N}^{(\beta)})^{-1}$, one has that the eigenvalue PDF is given by (\ref{1.2}) with $\lambda_1 = (\beta/2) (n_1 - N + 1) - 1$, $\lambda_2 = (\beta/2) (n_2 - N + 1) - 1$; see e.g.~\cite[Prop.~3.6.1]{Fo10}. 
In the real case ($\beta = 1$) the random matrices $W_{n_1,N}^{(\beta)}$ and $Y_{n_1,n_2,N}^{(\beta)}$ are fundamental to classical multivariate statistics, being generalisations of the univariate gamma and $F$ distributions respectively \cite{Mu82}. Complex Wishart matrices (i.e.~random matrices $W_{n_1,N}^{(\beta)}$ for $\beta = 2$) are prominent in the theory of wireless communications in modeling Rayleigh fading of signal~\cite{TV04}. Real Wishart matrices, although comparatively less common in this context, do find applications in modeling ``one-sided Gaussian" fading of communication signals~\cite{KP2010}. Moreover, the fixed (unit) trace variant of the Wishart matrices ($W_{n_1,N}^{(\beta)}/\mathrm{Tr} W_{n_1,N}^{(\beta)})$) serves as a model for random mixed quantum states within the bipartite formalism in quantum information theory~\cite{Fo10,ZS2001}. In addition to these, for each of $\beta=1,2$ and 4 the eigenvalue statistics of $W_{n_1,N}^{(\beta)}$ and $Y_{n_1,n_2,N}^{(\beta)}$ have applications in quantum transport problem \cite{Be97}.

Distinct from the interpretation of (\ref{1.1}) and (\ref{1.2}) as eigenvalue PDFs is their appearance in the theory of quantum many body systems in one-dimension of the Calogero-Sutherland type. Thus in the full parameter range their are many body Schr\"odinger operators, with one and two-body interactions only and the latter involving the inverse distance squared, for which the PDFs $\mathcal P^{\rm L}$ and $\mathcal P^{\rm J}$ are the square of the ground state wavefunction \cite{BF97a}, \cite[Ch.~11]{Fo10}.

In this article we return to a problem first addressed in the early days of the appearance of (\ref{1.1}) and (\ref{1.2}) with $\beta =1$ in multivariate statistics, but extended now to include all positive integer $\beta$. Starting with these joint eigenvalue PDFs, we ask about the marginal PDF of the individual eigenvalues, taken to be suitably ordered. The $\beta =1$ case of this problem for (\ref{1.2}) (after a fractional linear transformation so that the support is $x_l > 0$) was studied as long ago as 1945 by Roy \cite{Ro45}, with further refinements of this particular method, involving reduction formulas of so-called pseudo determinants, forthcoming in \cite{PD69} and \cite{Ek74}. A differential-difference recurrence method (see Section \ref{S2.2} below) was introduced by Davis around 1970 \cite{Da68,Da72a,Da72b}, leading to reduction formulas which in \cite{Ek74} were shown to be equivalent to an implementation of the method initiated by Roy. Moreover in \cite[Appendix A]{Da72b} a listing of the corresponding exact functional forms of the marginal eigenvalue PDFs for (\ref{1.2}) is given up to $N=5$. These involve products of up to $[N/2]+1$ incomplete beta functions. Soon after an explanation for such functional forms, in the context of an overarching Pfaffian structure, were obtained in the work of Krishnaiah and collaborators \cite{KC71,KW71}. These works used the method of integration over alternate variables, as introduced into random matrix theory by Mehta \cite{Me67}. The review \cite{Pi76} gives more details and context relating to this early work.

One notes that the scaling and limiting $x_l \mapsto \beta x_l \lambda_2/2$,
$\lambda_2 \to \infty$ in (\ref{1.2}) gives (\ref{1.1}). This allows results obtained for distributions associated with (\ref{1.2}) to be used to deduce the corresponding distributions for (\ref{1.1}). In \cite[Appendix B]{Da72b} there is a listing of the exact functional forms of the marginal eigenvalue PDFs for the largest eigenvalue of (\ref{1.1}) with $\beta = 1$ up to $N=5$, obtained using this the  differential-difference recurrence method from knowledge of the corresponding functional forms in relation to (\ref{1.2}).  These involve products of up to $[N/2]+1$ incomplete gamma functions. Again, this is consistent with the Pfaffian structure found subsequently in \cite{KC71,KW71}.

Later the calculation of the marginal eigenvalue PDFs, $\{ f_N(n;x) \}_{n=1,\dots,N}$ say, of the individual eigenvalues as specified by (\ref{1.1}) and (\ref{1.2}), received much attention in the case $\beta = 2$ \cite{TW94c,FW01a,FW04,Bo09}. Specifically, these works consider the probabilities $\{ E_N(n;(x,b)) \}_{n=0,\dots,N}$ that there are exactly $n$ eigenvalues in the interval $(x,b)$, with $b=\infty$ (Laguerre case) and $b=1$ (Jacobi case). The generating function
\begin{equation}\label{E1}
\Xi_N(z;x) := \sum_{n=0}^N z^n E_N(n;(x,b))
\end{equation}
for $\beta = 2$ has the special Fredholm determinant structure
\begin{equation}\label{E2}
\Xi_N(1-z;x) = \det ( \mathbb I - z K_{N,x}),
\end{equation}
where $\mathbb I$ is the identity operator, and $K_{N,x}$ is the symmetric finite rank projection operator onto the span of
$\{ (w(x))^{1/2} \phi_p(x) \}_{p=0,\dots,N-1}$, restricted to $(x,b)$; see e.g.~\cite[Ch.~9]{Fo10}. Here $w(x) = x^{\lambda_1} e^{-x}$ (Laguerre case), $w(x) = x^{\lambda_1}(1 - x)^{\lambda_2}$ (Jacobi case), with $\{ \phi_p(x) \}$ the  orthogonal polynomials corresponding to these weights. The works \cite{TW94c,FW01a,FW04} relate the generating function $\Xi_N(x;z)$ for general $\lambda_1, \lambda_2 > -1$ to Painlev\'e transcendents, while \cite{Bo09} shows how to use the Fredholm determinant form directly for the efficient and high precision evaluation of $\{ E_N(n;(x,b)) \}$ and $\{ f_N(n;x) \}$. Although not the topic of the present work, a noteworthy feature of analysis based on (\ref{E2}) is the ease at which the scaled (soft or hard edge as appropriate) $N \to \infty$ limit can be taken; see the references cited above. The finite $N$ determinantal structure was used more directly in the numerical evaluation of $\{ f_N(n;x) \}$ for small $N$ in \cite{ZCW09}, the results of which were moreover put in a wireless communications context.

Essential to our work is
the fact much simpler functional forms associated with the marginal PDF of the individual eigenvalues (\ref{1.1}) and (\ref{1.2}) than evident from the listings in \cite{Da72b},  or the Painlev\'e transcendent expressions of \cite{TW94c,FW01a,FW04}, and that apply for all $\beta$ a positive integer provided $\lambda_1$ is a non-negative integer. Thus denote by $f_N^{(\cdot)}(n;x)$ the PDF of the $n$-th largest eigenvalue as implied by (\ref{1.1}) ($(\cdot) = \rm L$) or (\ref{1.2}) ($(\cdot) = \rm J$). Define polynomials
\begin{equation}\label{qq}
q^{(\cdot)}(x;n,j)  = \sum_{k=0}^{j(\lambda_1+(N-j)\beta)} \alpha_{jk}^{(\cdot)} x^k,  
\end{equation}
for some coefficients $ \alpha_{jk}^{(\cdot)}$.  Then the simplified functional forms are
\begin{align}
f_N^{\rm L}(n;x) & = \sum_{j=n}^N
e^{-\beta j x/2} q^{\rm L}(x;n,j) 
\label{4a}  \\
f_N^{\rm J}(n;x) & = \sum_{j=n}^N
(1 - x)^{j(\lambda_2 + 1) + j (j-1) \beta/2 - 1} q^{\rm J}(x;n,j).
\label{4b}
\end{align}

In the case $\beta =1$ and $n = N$ of (\ref{4a}) this functional form can be found in the work of Edelman \cite{Ed91}, along with a recursive procedure to compute the polynomials $q^{\rm L}(x;N,j)$. In fact for $n=N$ and with $\lambda_1$ a non-negative integer the structures (\ref{4a}) and (\ref{4b}) are valid for general $\beta >0$; in this setting of the Laguerre case see \cite{Fo93c} for the evaluation of $q^{\rm L}(x;N,j)$ in terms of certain generalised hypergeometric functions based on Jack polynomials. For $n=0$ and $\beta = 2$ of $f_N^{\rm L}(n;x)$ these structures were pointed out in \cite{DMJ03}. Recursive procedures to compute $f_N^{\rm L}(n;x)$ and $f_N^{\rm J}(n;x)$ for $n=0$, applicable for $\lambda_1$ a non-negative integer and $\beta$ a positive integer, along with software for their implementation, have been given recently in \cite{FK19,FK23}. The primary goal  in the present paper is to extend the latter two studies to all of the marginal PDFs, not just $n=0$, and so give a recursive computation of the polynomials (\ref{qq}). 

The method we use to compute the polynomials (\ref{qq}) is detailed 
in Section \ref{S2}. It is based on the general $\beta > 0$ extension \cite{Fo93,FR12,FT19,Ku19} of the differential-difference equation first identified by Davis \cite{Da72a}, and a recursive procedure which can be found in a further work of Davis \cite{Da72b}. Crucial to our work, and distinct from that of Davis, is the simplicity of the transcendental functions for which the polynomials are coefficients, as given in (\ref{4a}) and (\ref{4b}), which results from the assumption $\lambda_1 \in \mathbb Z_{\ge 0}$. This allows for the implementation of the scheme using Mathemetica code, which is provided as part of this work.

In Section \ref{S3} we identify special functional forms associated with $\{f_N^{\rm L}(n;x)\}$ in the case that $\lambda_1+1/2$ is a non-negative integer, specialising (for the most part) to $\beta = 1$. There are at least two motivations for this. One is that $\lambda_1+1/2$ a non-negative integer results in (\ref{1.1}) whenever the matrix size of the rectangular standard Gaussian matrix $X_{n_1,N}^{(1)}$ underlying the real Wishart matrix $W_{n_1,N}^{(1)} = (X_{n_1,N}^{(1)})^T X_{n_1,N}^{(1)}$ is such that $n_1 - N$ is even. Another is the relationship between the Wishart and fixed trace Wishart ensembles, which calls for taking the inverse Laplace transformation of the marginal distributions of the former. This step can only be carried out at an analytic level if it possible to expand the marginal distributions of (\ref{1.1}) in a sufficiently simple form. Taking $n=1$ and $N$ even for for definiteness, our main finding (based on a recursive implementation of the differential-difference equation, but not proven in general) is the validity of the expansion 
for the cumulative distribution function (CDF) of the largest eigenvalue (Eq.~(\ref{B2})) below)
$$
F_N^{\rm L}(1;x) =
\sum_{l=1}^{N/2+1} \Big (  e^{-(l-1)x} p_{l,1}^{\rm e}(x) +
 \sqrt{x} {\rm erf}(\sqrt{x/2}) e^{-(l-1/2)x}
p_{l,2}^{\rm e}(x) \Big ),
$$
with
$p_{l,1}^{\rm e}(x), p_{l,2}^{\rm e}(x)$ polynomials. An explicit computation of the inverse Laplace transform is indeed possible, allowing us to give as an application the exact functional form of certain Landauer conductances \cite{KP11a}, which has been implemented as Mathematica~\cite{Mtmk} code.

\section{The recursive scheme}\label{S2}
\subsection{The marginal PDFs and inter-relations}
For a general eigenvalue PDF $\mathcal P(x_1,\dots,x_N)$ symmetric in $\{ 
x_i \}$ and supported on $x_i \in (a,b)$, one notes that the PDF of the $n$-th $(n=1,\dots,N)$ largest eigenvalue $f_N(n;x)$ can be written
\begin{equation}\label{5.0a}
f_N(n;x) = n C_n^N \int_{R_{n}}
\mathcal P(x,x_2,\dots,x_N) \, dx_2 \cdots dx_N,
\end{equation}
where $C_n^N$ denotes the binomial coefficient and $R_{n}$ denotes the domain of integration
\begin{equation}\label{Rn}
b > x_2,\dots, x_n > x > x_{n+1},\dots, x_N > a.
\end{equation}
From the PDF $f_N(n;x)$ of the $n$-th largest eigenvalue, the corresponding cumulative distribution $F_N(n;x)$ is computed according to
\begin{equation}\label{5.1c}
F_N(n;x) := \int_0^x f_N(n;y) \, dy.
\end{equation}

Closely related to $f_N(n;x)$ is the quantity $E_N(n;(x,b))$, defined as the probability that there are exactly $n$ eigenvalues in $(x,b)$. As a multiple integral, one has
\begin{equation}\label{5.1b}
E_N(n;(x,b)) =  C_n^N \int_{\tilde{R}_{n}}
\mathcal P(x_1,\dots,x_N) \, dx_1 \cdots dx_N,
\end{equation}
where $\tilde{R}_{n}$ denotes the domain of integration
\begin{equation}\label{Rn1}
b > x_1,\dots, x_n > x > x_{n+1},\dots, x_N > a.
\end{equation}
Comparison of (\ref{5.0a}) and (\ref{5.1b}) shows that the two are related by
\begin{equation}\label{5.1c1}
f_N(n;x) = {d \over dx} E_N(n-1;(x,b)) + f_N(n-1;x),
\end{equation}
with $f_N(0;x)=0$. In terms of the cumulative distribution function (\ref{5.1c}), this last equation can equivalently be written
\begin{equation}\label{5.1c2}
E_N(n-1;(x,b)) = F_N(n;x) - F_N(n-1;x) \quad (n=1,\dots,N+1),
\end{equation}
with
$$
F_N(0;x) := 0, \qquad F_N(N+1;x) := 1.
$$

Specialise now to the class of PDFs of the form
\begin{equation}\label{5.0b}
{\mathcal P}(x_1,\dots,x_N) = {1 \over Z_N}
\prod_{l=1}^N w(x_l) \prod_{1 \le j < k \le N} | x_k - x_j|^{\beta},
\end{equation}
for some weight function $w(x)$, as is the structure of (\ref{1.1}) and (\ref{1.2}). With $\beta$ a positive integer, we can then write (\ref{5.0a}) in the particular form
\begin{equation}\label{5.1X}
f_N(n;x) = n C_n^N (-1)^{\beta (n-1)} {w(x) \over Z_N} \int_{R_{n}} \prod_{l=2}^N w(x_l) (x - x_l)^\beta \prod_{2 \le j < k \le N} |x_k - x_j|^\beta
 \, dx_2 \cdots dx_N.
\end{equation}
Similarly in relation to the probability (\ref{5.1b}), the specialisation (\ref{5.0b}) gives the particular form
\begin{equation}\label{5.1Y}
E_N(n;(x,b)) =   {C_n^N \over Z_N} \int_{\tilde{R}_{n}} \prod_{l=1}^N w(x_l) \prod_{1 \le j < k \le N} |x_k - x_j|^\beta
 \, dx_1 \cdots dx_N.
\end{equation}
Most important in our ability to recursively compute the polynomials (\ref{qq}) is the functional quantity associated with (\ref{5.1Y})
\begin{equation}\label{5.2a}
E_N(n;(x,b))[f] =   {C_n^N \over Z_N} \int_{\tilde{R}_{n}} f(x_1,\dots,x_N) \prod_{l=1}^N w(x_l)  \prod_{1 \le j < k \le N} |x_k - x_j|^\beta
 \, dx_1 \cdots dx_N.
\end{equation}
Note in particular, by comparing with (\ref{5.1Y}), that we have for $\beta$ a positive integer
\begin{equation}\label{5.2b}
(-1)^{\beta n} (N+1) {Z_N \over Z_{N+1}} w(x) E_N(n;(x,b))\Big [\prod_{l=1}^N (x - x_l)^\beta \Big ] =
f_{N+1}(n+1;x).
\end{equation}

\subsection{Differential-difference recurrences from the theory of Selberg integrals}\label{S2.2}

Related to the LHS of (\ref{5.2b}) in the Laguerre and Jacobi cases are multiple integrals central to the broader theory of the Selberg integral \cite[Ch.~4]{Fo10}.
Let $e_p(y_1,\dots,y_N) $ denote the elementary symmetric polynomials in $\{ y_j \}_{j=1}^N$, and
define in the Jacobi case
\begin{multline}\label{5.1J}
J_{p,N}^{(\alpha)}(x) = {1 \over C_p^N} \int_{\tilde R_n} 
\prod_{l=1}^N t_l^{\lambda_1} (1 - t_l)^{\lambda_2} (x - t_l)^\alpha \\
\times \prod_{1 \le j < k \le N} | t_k - t_j|^\beta 
e_p(x - t_1,\dots, x - t_N) \,
dt_1 \cdots dt_N.
\end{multline}
 This family of multiple integrals satisfies the
differential-difference system \cite{Fo93}, \cite[\S 4.6.4]{Fo10},
later observed to be equivalent to a certain Fuchsian matrix differential equation \cite{FR12},
\begin{multline}\label{5.1aJ}
(N - p) E_p^{\rm J} J_{p+1}(x) 
= (A_p^{\rm J} x + B_p^{\rm J}) J_p(x) - x(x-1) {d \over dx} J_p(x) + D_p^{\rm J} x ( x - 1) J_{p-1}(x),
\end{multline}
where we have abbreviated $J_{p,N}^{(\alpha)}(x) =: J_p(x)$, and
\begin{align*}
A_p^{\rm J} & = (N-p) \Big ( \lambda_1 + \lambda_2 + \beta (N - p - 1) + 2(\alpha + 1) \Big ) \\
B_p^{\rm J} & = (p-N)  \Big ( \lambda_1 + \alpha + 1 + (\beta/2) (N - p - 1) \Big ) \\
D_p^{\rm J} & = p \Big ( (\beta/2)(N-p) + \alpha + 1 \Big ) \\
E_p^{\rm J} & = \lambda_1 + \lambda_2 + 1 + (\beta/2) (2N - p - 2) + (\alpha + 1).
\end{align*}
Note that the coefficients are all independent of $n$ as occurs in the domain of integration of (\ref{5.1J}) and therefore so too is (\ref{5.1aJ}); this parameter only enters via the initial condition $J_0(x)$. We remark too that the case $\beta =1$ of (\ref{5.1aJ}) is equivalent to the differential-difference recurrence obtained in the pioneering work of Davis \cite{Da72a}. Another point of interest is that the family of integrals $\{ J_{p,N}^{(\alpha)}(x) \}$ satisfy a certain matrix difference equation in the parameter $\lambda_1$ \cite{FI10a}.

Similarly, in the Laguerre case define
\begin{equation}\label{2.15}
L_{p,N}^{(\alpha)}(x)  = {1 \over C_p^N}  \int_{\tilde R_n} \prod_{l=1}^N t_l^{\lambda_1} e^{-\beta t_l/2} |x-t_l|^{\alpha}\prod_{1\leq j<k \leq \nu} | t_k -t_j|^{\beta} e_p(x-t_1,\ldots,x-t_N).
\end{equation}
	 Analogous to (\ref{5.1aJ}) this family of integrals satisfies the differential-difference system \cite{FT19}
\begin{equation}	
\label{recur}
(\beta/2) (N-p) L_{p+1,N}^{(\alpha)}(x) = ((\beta/2)(N-p)x+B_p^{\rm L})L_{p,N}^{(\alpha)}(x) 
	+ x{d\over dx}L_{p,N}(x) - D_p^{\rm L} x L_{p-1,N}^{(\alpha)}(x),
\end{equation}
where $p=0,1,...,N-1$ and
	$$
		B_p^{\rm L} = (p-N)[\lambda_1 + \alpha +1+ (\beta/2)(N-p-1)],\qquad
		D_p^{\rm L} = p[(\beta/2)(N-p)+\alpha+1].
	$$
 As noted in \cite{FT19}, (\ref{2.15}) can be deduced as a limiting case of (\ref{5.1aJ}); recall the first sentence of the paragraph above that containing (\ref{qq}).

With knowledge of (\ref{2.15})
and (\ref{5.1aJ}), as already observed and implemented in \cite{FK19}, \cite{FK23}, it is straightforward in theory at least to inductively compute $\{f_N(1;x) \}$ in the Laguerre and Jacobi cases with $\beta$ a positive integer. Thus
(\ref{5.1c2}) in the case $n=1$ reads
\begin{equation}\label{5.1c3}
E_N(0;(x,b)) = F_N(1;x) = \int_0^x f_N(1;x) \, dx.
\end{equation}
But with $N=1$, in the setting of (\ref{5.0b}), $f_1(1;(x,b)) = {1 \over Z_1}w(x)$, so we have knowledge of $E_1(0;(x,b))$. 

For general $N$, in the Laguerre and Jacobi cases and with $\beta$ a positive integer, the differential-difference recurrences can be used to compute $f_{N+1}(1;x)$ from knowledge of $E_{N}(0;(x,b))$. This is because the differential-difference recurrences can compute the functional of $E_{N}(0;(x,b))$ on the LHS of (\ref{5.2b}) from knowledge of $E_{N}(0;(x,b))$, and furthermore the normalisations $Z_N$ are known from (\ref{1.3a}) and (\ref{1.3b}). Consider the Laguerre case for definiteness. Thus one sees that $L_{0,N}^{(0)}(x)$ is by definition proportional to $E_{N}^{\rm L}(0;(x,\infty))$. Iterating the recurrence for $p=0,1,\dots,N-1$ gives $L_{N,N}^{(0)}(x)$, which from the definitions is equal to $L_{0,N}^{(1)}(x)$. Repeating this iterative procedure $\beta$ times gives the sought functional on the LHS of (\ref{5.2b}). 

Having started with $N=1$, we now have knowledge of $f_2(1;x)$. Integrating as required in (\ref{5.1c3}) we now know $E_2(0;(x,b))$. The above procedure can now be repeated to arrive at $f_3(1;x)$, etc. Thus the method involves a double recurrence. 

Taking a practical viewpoint, one sees that the need to compute the integral in (\ref{5.1c3}) at each step of $N$ is a severe limitation in the general case. However, as observed in \cite{FK19}, \cite{FK23}, if one restricts to the special case of $\lambda_1$ a non-negative integer, this limitation can be removed. This is due to the integration formulas
\begin{equation}\label{7.1b}
\int_0^x s^r e^{-\beta m s/2} \, ds = r! \Big ( {2 \over \beta m} \big )^{r+1} \bigg (
1 - e^{-\beta m x/2} \sum_{k=0}^r {1 \over k!} \Big ( {\beta m x \over 2} \Big )^k \bigg ), \quad r \in \mathbb Z_{\ge 0}
\end{equation}
and
\begin{equation}\label{7.1c}
\int_0^x y^r(1-y)^b \, dy =
\tilde{B}(r,b) - \sum_{p=0}^r (-1)^p C_p^r\, {1 \over b + p + 1} (1 - x)^{b+p+1}, \quad r \in \mathbb Z_{\ge 0},
\end{equation}
with
$$
\tilde{B}(r,b) = \int_0^1 y^r(1-y)^b \, dy = {\Gamma(r+1) \Gamma(b+1) \over \Gamma(r+b+2)}.
$$
Alternatively, if $b\in\mathbb{Z}_{\ge 0}$, instead of \eqref{7.1c}, we can use
\begin{align}\label{7.1c-alt}
\int_0^x y^r(1-y)^b \, dy = \sum_{p=0}^b (-1)^p C_p^b\, \frac{1}{r+p+1}x^{r+p+1}.
\end{align}

The computational codes provided with \cite{FK19}, \cite{FK23} for the computation of $\{f_N(1;x) \}$ in the Laguerre and Jacobi cases with $\beta$ a positive integer and $\lambda_2$ a non-negative integer makes essential use of these integration formulas.

It has previously been noted in \cite{FK19} and \cite{FK23} that the theory of the above paragraph implies the expansions (\ref{4a}) and (\ref{4b}) with $n=1$, up to the precise degree of the polynomials (\ref{qq}). The latter was determined in the cited references by a distinct argument; see Remark \ref{R1} below. 

\subsection{Recursive computation of $\{f_N(n;x) \}$ for $n \ge 2$}\label{S2.3}

A theoretical recursive scheme for the computation of $\{f_N(n;x) \}$ for $n \ge 2$ using the differential-difference recurrences of the previous subsection has already been identified by Davis \cite{Da72b}. In short, there is now need for an induction in $n$, based on (\ref{5.1c1}). But as in the case $n=1$ this was not particularly practical in the general case, with again the need to repetitively compute the integral
in (\ref{5.1c3}) in both each step of $N$ and each step in $n$. But we now know how this problem can be overcome by specialising to $\lambda_1$ a non-negative integer, making it timely to revisit this problem.

Let us then consider the induction in $n$ identified in \cite{Da72b} in more detail. As part of the computation of $\{f_N(1;x) \}$ revised above, we have knowledge of the corresponding sequence of cumulative distributions $\{F_N(1;x) \}$. According to (\ref{5.1c2}) with $n=2$, $N=1$
\begin{equation}\label{u91}
E_1(1;(x,b)) = 1 - F_1(1;x).
\end{equation}
But from the definitions $E_1(1;(x,b))$ is proportional to $L_{0,N}^{(0)}(n;N) |_{n=N=1}$ (Laguerre case) and $J_{0,N}^{(0)}(n;N) |_{n=N=1}$ (Jacobi case). The differential-difference recurrences (\ref{2.15})
and (\ref{5.1aJ}) can be used to compute from this $f_2(2;x)$, and then by integrating using (\ref{7.1b}) or (\ref{7.1c}) the cumulative distribution $F_2(2;x)$ follows according to (\ref{5.1c}). Making use of the latter in  (\ref{5.1c2}) with $n=2$, $N=2$, together with knowledge of $F_2(1;x)$ evaluates $E_2(1;(x,b))$. Continually repeating this procedure gives $\{f_N(2;x) \}_{N=2,\dots,N^*}$ and $\{F_N(2;x) \}_{N=2,\dots,N^*}$, where $N^*$ is the desired maximum value of $N$.

We move on now to the $n=3$ case, beginning with the analogue of (\ref{u91})
$$
E_2(2;(x,b)) = 1 - F_2(2;x),
$$
and repeat the same steps as for the $n=2$ case to arrive at the evaluations of $\{f_N(3;x) \}_{N=3,\dots,N^*}$ and $\{F_N(3;x) \}_{N=3,\dots,N^*}$. Iterating in $n$, we can compute in this way $\{ f_N(n;x) \}_{N=n,\dots,N^*}$ for each
$n=2,\dots,N^*$.

A different order of computation is also possible. This is to first compute $\{ f_N(N;x) \}_{N=1,\dots,N^*}$. The point here is that according to (\ref{5.1c2}) we have
$$
E_N(N;(x,b)) = 1 - F_N(N;x) \qquad (N=1,\dots,N^*).
$$
Thus the iterative procedure can be initiated with knowledge of $F_1(1;x)$ only to compute $\{ f_N(N;x) \}_{N=1,\dots,N^*}$. This fact, with $f_N(N;x)$ interpreted as the PDF of the smallest eigenvalue, has been observed previously and made into a practical algorithm \cite{Ed91,Ku19,FT19,FK23}. With $\{ f_N(N;x) \}_{N=1,\dots,N^*}$ and thus upon integration $\{ F_N(N;x) \}_{N=1,\dots,N^*}$ thus known, we use (\ref{5.1c2}) with $n=N=2$ to initiate the computation of $\{ f_N(N-1;x) \}_{N=2,\dots,N^*}$, etc. In this ordering, last to be computed is $f_{N^*}(1;x)$.

The two distinct recurrence orderings are evident from the following diagram, in which the arrows indicate the necessary inputs to arrive at each $f_N(n;x)$; all entries bar those along the first row and the leading diagonal have two incoming arrows.

$$
\begin{array}{ccccccccccc}
f_1(1;x) & \rightarrow & 
f_2(1;x) & \rightarrow &
f_3(1;x) & \rightarrow & f_4(1;x)  &  \rightarrow & \cdots &
\rightarrow & f_{N^*}(1;x) \\
& \searrow & & \searrow & & \searrow &&  \searrow & \cdots & \searrow & \\
&& f_2(2;x) & \rightarrow &
f_3(2;x) & \rightarrow & f_4(2;x)   &  \rightarrow & \cdots &
\rightarrow & f_{N^*}(2;x) \\
&& & \searrow & & \searrow & & \searrow & \cdots & \searrow & \\
&&&& f_3(3;x) & \rightarrow &
f_4(3;x) &  \rightarrow & \cdots &
\rightarrow & f_{N^*}(3;x) \\
&& &&& \searrow & & \searrow &  \cdots & \searrow & \\
&& &&& & \ddots & &   \ddots   &\vdots & \\
&&&&&  &  & & \cdots &
\rightarrow &  f_{N^*}(N^*-1;x) \\
&& &&& &  & &      &\searrow & \\
&&&&&  &  & &  & &
  f_{N^*}(N^*;x)
\end{array}
$$

\begin{remark}\label{R1}
The exact degree of the polynomials is not evident from this procedure. On this point, consider for definiteness the Laguerre case. After a simple change of variables, we have from (\ref{5.1X})
\begin{multline}
f_N(n;x) \propto x^{\lambda_1} e^{-\beta x/2} x^{(N-1) (1 + \lambda_1 + \beta N/2)}
\int_1^\infty dx_2 \cdots \int_1^\infty dx_n \int_0^1 dx_{n+1} \cdots \int_1^\infty dx_N \, \\
\times \prod_{l=2}^N x_l^{\lambda_1}
e^{-\beta N x_l/2} |1 - x_l|^\beta\prod_{2 \le j < k \le N} |x_k - x_j|^\beta.
\end{multline}
For $\lambda_1$ a non-negative integer and $\beta$ a positive integer, and upon appropriate ordering of $\{ x_j \}_{j=2}^n$ and
$\{ x_j \}_{j=n+1}^N$ for $\beta$ odd, one sees that the integrand is equal to $\prod_{l=2}^N x_l^{\lambda_1}
e^{-\beta N x_l/2}$ times a polynomial. Using this fact, following \cite[Appendix A]{FK19}, one sees that the expansion (\ref{4a}) with polynomials of maximum degree as given in (\ref{qq}) results by an evaluation of the multidimensional integral upon repeated integration by parts bases on the exponential factors.
\end{remark}

Based on the recipe outlined in Sections~\ref{S2.2} and~\ref{S2.3}, we provide Mathematica~\cite{Mtmk} codes in the supplementary material capable of computing exact analytical forms for the marginal distributions of all ordered eigenvalues, indexed by $n=1,2,...,N^*$, for a given $N^* (\ge 2)$. For the Laguerre case, the code {\sf "Laguerre-A"} can handle the case of $\lambda_1$ a non-negative integer and $\beta$ a positive integer. It is based on the use of the integral formula given in~\eqref{7.1b}. The case of $\lambda_1+1/2$ a non-negative integer is also considered in other codes, and is described in Section~\ref{S3.2} below. For the Jacobi case, the code {\sf "Jacobi-A"} can be employed for the case $\lambda_1>-1$, $\lambda_2$ a non-negative integer, and $\beta$ a positive integer. It makes use of the integration formula~\eqref{7.1c-alt}. For the case $\lambda_1$ a non-negative integer, $\lambda_2>-1$ , and $\beta$ a positive integer, while the integration formula~\eqref{7.1c} may be used in the recursive procedure, it is convenient in Mathematica to instead implement the swap $\lambda_1\leftrightarrow \lambda_2$, obtain the results using~\eqref{7.1c-alt}, and then consider $x\to1-x$. This still produces the marginal distributions of all eigenvalues, albeit in the reverse order. This strategy is utilized in the code {\sf "Jacobi-B"}.

In Figs.~\ref{FigLagA} and~\ref{FigJacAB}, we present example plots illustrating the marginal distributions of eigenvalues for Laguerre and Jacobi ensembles, respectively. These analytical formula-based curves agree perfectly with matrix-model-based simulation results, which are omitted here for clarity. However, these comparisons are available in the examples provided within the Mathematica files.

\begin{figure*}[!t]
  \centering
\includegraphics[width=0.95\textwidth]{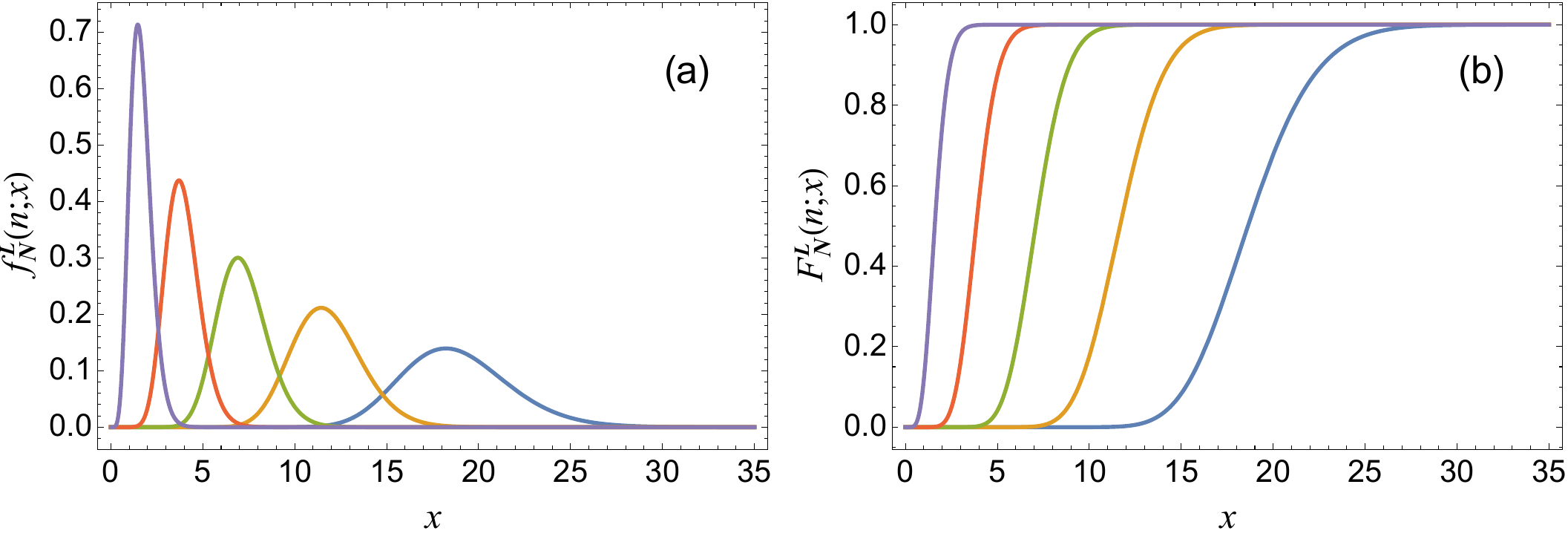}  \caption{Plots of (a) marginal PDFs $f_N^{\rm L}(n;x)$ and (b) marginal CDFs $F_N^{\rm L}(n;x)$ of eigenvalues for Laguerre case with $N=5, \beta=3, \lambda_1=6$. In both panels, the curves from right to left correspond to $n=1$ (largest) to $n=N$ (smallest) eigenvalues.}
\label{FigLagA}
\end{figure*}
\begin{figure*}[!h]
  \centering
\includegraphics[width=0.95\textwidth]{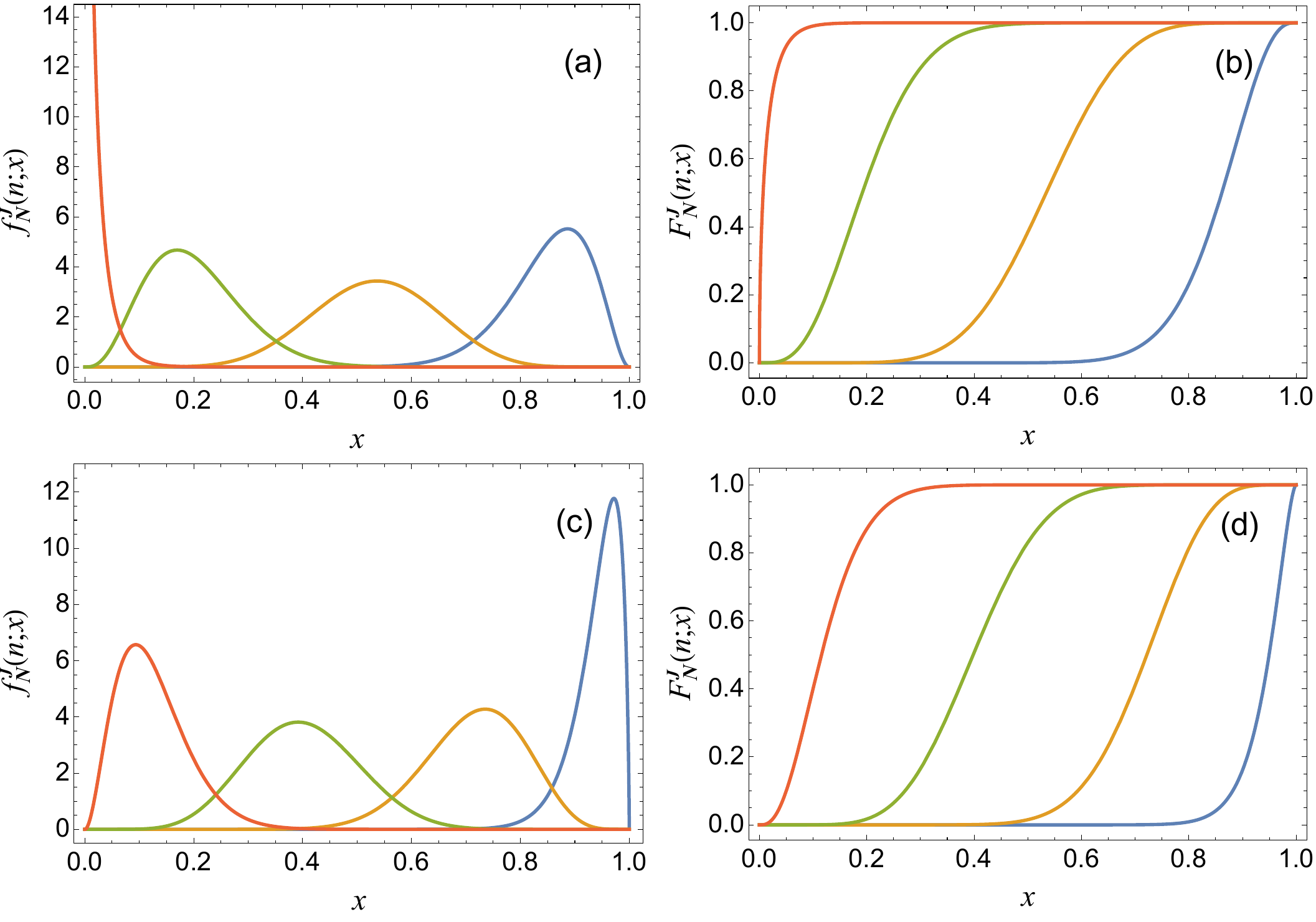}  \caption{Plots of marginal PDFs $f_N^{\rm J}(n;x)$ and marginal CDFs $F_N^{\rm J}(n;x)$ for Jacobi ensemble. Panels (a) and (b) are for parameter choices $N=4,\beta=3,\lambda_1=-1/2,\lambda_2=2$, whereas (c) and (d) are for $N=4,\beta=3, \lambda_1=2, \lambda_2=3/4$. In all plots, the curves from right to left correspond to $n=1$ (largest) to $n=N$ (smallest) eigenvalues.}
\label{FigJacAB}
\end{figure*}

\subsection{Relationship to the eigenvalue density}
The eigenvalue density, $\rho_N(x)$ say, is specified in terms of the eigenvalue PDF according to
\begin{equation}\label{D1}
\rho_N(x) = N \int_a^b dx_2 \cdots \int_a^b dx_N \, \mathcal P(x,x_2,\dots,x_N).
\end{equation}
For the class of PDFs (\ref{5.0b}),
(\ref{D1}) reads
\begin{equation}\label{D2a}
\rho_N(x) = {N \over Z_N} w(x) \int_a^b dx_2 \,  \cdots \int_a^b dx_N \, 
\bigg ( \prod_{l=1}^N w(x_l) |x - x_l|^\beta \bigg )
\prod_{2 \le j < k \le N}|x_k - x_j|^\beta.
\end{equation}
In particular, for $\beta$ even the multiple integral is a polynomial in $x$ of degree $\beta (N - 1)$.
The computation of this polynomial for the Laguerre $\beta$ ensemble using the differential-difference recurrence (\ref{recur}) has been exhibited in \cite{FT19}. Analogous use of (\ref{5.1aJ}) provides a computational scheme of this polynomial in the case of the Jacobi $\beta$ ensemble.

On the other hand, the eigenvalue density relates to the marginal eigenvalue PDFs according to the well known sum rule (see e.g.~\cite[Eq.~(8.18)]{Fo10})
\begin{equation}\label{D2}
\rho_N(x) = \sum_{n=1}^N f_N(n;x),
\end{equation}
as follows upon comparing (\ref{D1}) to (\ref{5.0a}). Comparison of the factors of $w(x)$ in (\ref{D2a}) and (\ref{4a}), (\ref{4b}) tells us that for $\beta$ even
\begin{equation}\label{D2b}
\rho_N(x) = \begin{cases}
e^{-\beta x/2} q^{\rm L}(x;1,1) & {\rm Laguerre}, \\
(1-x)^{\lambda_2} q^{\rm J}(x;1,1) & {\rm Jacobi}. \end{cases}
\end{equation}
We remark that polynomial factors in the density in the cases of $\beta =2$ and $\beta = 4$ can be expressed in terms of the underlying orthogonal polynomials; see \cite[Eqns.~(5.13) and (6.59) respectively, the latter with $x=y$]{Fo10}, thus providing for evaluations of $q^{(\cdot)}(x;1,1)$.

\section{More on exact distributions in the Laguerre case with $\beta =1$}\label{S3}

\subsection{The fixed trace variant}
As pointed out in the introduction, relevant to applications in quantum information --- specifically bipartite entanglement problems \cite{LP88} --- is the unit trace variant of the Laguerre $\beta$-ensemble. For this (\ref{1.1}) takes the modified form
\begin{equation}\label{1.1F}
    \mathcal P^{\rm fL}(x_1,\dots,x_N) = {1 \over Z_N^{\rm fL}}
    \prod_{l=1}^N x_l^{\lambda_1} \delta \Big ( \sum_{l=1}^N x_l - 1 \Big ) \chi_{x_l > 0}
    \prod_{1 \le j < k \le N} |x_k - x_j|^\beta.
\end{equation}
For $\beta = 1,2,4$ and $\lambda_1 = \beta (n_1-N + 1)/2 -1$, this can eigenvalue PDF can be realised as the eigenvalue PDF of
the scaled Wishart matrix $W_{n_1,N}^{(\beta)}/{\rm Tr} \,( W_{n_1,N}^{(\beta)})$ \cite{GN99},
\cite[Exercises 3.3 q.3]{Fo10}.

One observes that taking the Laplace transform of the fixed trace PDF (\ref{1.1F}) gives back, up to normalisation and setting the transform variable $t=\beta/2$, the original PDF (\ref{1.1}). Conversely, upon scaling the eigenvalues of (\ref{1.1}) $x_l \mapsto 2 t/\beta$, the inverse Laplace transform can be taken to arrive at (\ref{1.1F}). Moreover, this prescription applies too at the level of the distribution functions, telling us that
\begin{equation}\label{3.2}
F_N^{\rm fL}(n;x) = {Z_N^{\rm L} \over Z_N^{\rm fL}} {x^{\gamma -1} \over 2 \pi i} \int_{c - i \infty}^{c + i \infty} {F_N^{\rm L}(n;2s/\beta) \over (2s/\beta)^\gamma} e^{s/x} \, ds, \quad c>0,
\end{equation}
where $\gamma := N(\lambda_1 + \beta (N-1)/2 + 1)$.
Significantly, knowledge of the analogue of the expansion (\ref{4a}) for $F_N^{\rm L}(n;x)$ allows the inverse Laplace transform to be computed explicitly, giving that
\cite[Eq.~(9)]{FK19}
\begin{equation}
F_N^{\rm fL}(n;x)  = \Gamma(\gamma)\sum_{j=n}^N
\Theta(1-j x) \sum_{k=0}^{j(\lambda_1+(N-j)\beta)} {\tilde{\alpha}_{jk}^{\rm L} \over \Gamma(\gamma - k )} \Big ( {2x \over \beta} \Big )^k (1 - j x)^{\gamma - k - 1},
\label{4aF}
\end{equation}
where $\Theta(z)$ denotes the Heaviside step function. Here $\tilde{\alpha}_{jk}^{\rm L}$ are the coefficients in the analogues of the polynomials $q^{\rm L}(x;n,j)$ as applies to the analogue of (\ref{4a}) for $F_N^{\rm L}(n;x)$. The explicit form of the expansion (\ref{4aF}) in the case $\beta =2$, $\lambda_1=0$ and with small values of $N$ has been the subject of the recent work \cite{SBL22}. For $n=1$ and $\beta = 2$ see too \cite{AKT19}.

It has been noticed in \cite{Vi11} that $F_N^{\rm fL}(1;x)$, i.e.~the cumulative distribution function of the largest eigenvalue in the Laguerre fixed trace ensemble, permits an alternative interpretation. To see this, one makes use of the definition of this cumulative distribution, together with the functional form (\ref{1.1F}) and a simple change of variables, to obtain
\begin{align}\label{3.3}
F_N^{\rm fL}(1;x) & = \int_0^x dx_1 \cdots \int_0^x dx_N \,
\mathcal P^{\rm fL}(x_1,\dots,x_N)
\nonumber \\
& = {Z_N^{\rm J}|_{\lambda_2=0} \over  Z_N^{\rm fL}} x^{\gamma - 1}
\int_0^1 dx_1 \cdots \int_0^1 dx_N \, \delta \Big ( \sum_{l=1}^N x_l - 1/x \Big )
\mathcal P^{\rm J}(x_1,\dots,x_N)|_{\lambda_2=0}.
\end{align}
The multiple integral in the second line of (\ref{3.3}) has the interpretation as the distribution of the trace (in the variable $1/x$) of the Jacobi $\beta$ ensemble with $\lambda_2 = 0$. In the cases $\beta = 1,2$ and 4, for which the Jacobi $\beta$ ensemble with $\lambda_2=0$ and $\lambda_1 = (\beta/2) (n_1-N + 1/\beta) - 1$ specifies the distribution of the transmission eigenvalues in a two-lead quantum dot \cite{Be97}, this quantity gives the distribution of the Landauer conductance $P_g(g)$ say, which is a physically accessible quantity; see \cite[\S 4]{FK19} and references therein. As previously done in \cite{FK19}, starting with the expansion (\ref{4a}), after having determined the coefficients $\alpha_{jk}^{\rm L}$ in $q^{\rm L}$ of (\ref{qq}) according to the recursive scheme of \S \ref{S2.2}, the equality implied by (\ref{3.3}) allows for the exact evaluation of $P_g(g)$ for $\beta = 1,2$ and $4$ provided the parameter $\lambda_1$ is a non-negative integer. With $\lambda_1 = (\beta/2) (n_1-N + 1) - 1$, this covers all physically interesting cases except $\beta =1$ and $\lambda_1 -1/2$ a non-negative integer.

This latter case has been the subject of some earlier literature. In particular, for small values of $n_1,N$ with $n_1 \ge N$, a tabulation has been given in \cite[\S 3.1]{KP11a}, revealing finite expressions analogous to (\ref{4aF}) again holding true. As some examples, with $n_1=N=3$ and thus $\lambda_1 = -1/2$ we read off from \cite[Eq.~(21)]{KP11a} that
\begin{equation}\label{3.4}
P_g(g) = {6 \over 7} g^{7/2} \chi_{0<g<1} + {3 \over 28} \Big ( 
35 g^3 - 175 g^2 + 273 g - 125 - 8 (g - 2)^{5/2} (g + 5) \Theta(g-2) \Big ) \chi_{1<g<3},
\end{equation}
while \cite[Eq.~(24)]{KP11a} 
gives that for $n=N=4$ and thus again $\lambda_1 = -1/2$ 
\begin{multline}\label{3.5}
P_g(g) = {5 \over 27456} \Big (
429 g^7 - 512 (g - 1)^{9/2}(6g^2 - 64 g + 201) \Theta(g-1) \Big )
 \chi_{0<g<2} \\
- {5 \over 27456} \Big ( 
429 g^7 - 72072 g^5 + 
672672 g^4 - 2800512 g^3 + 6150144 g^2 - 6935552g \\
+ 3158016 - 1024 (g-3)^{11/2} (3 g + 4) \Theta(g-3) \Big ) \chi_{2<g<4}.
\end{multline}
Similar explicit functional forms are given in \cite{KP11a} in the cases $N=3$ and $N=4$ with $\lambda_1 = 1/2$ as well.

The finite functional forms (\ref{3.4}) and (\ref{3.5}), together with the relation between $P_g(g)$ and $f_N^{\rm L}(1;x)$ as revised above, gives motivation for investigating is analogous finite functional forms carry over to $f_N^{\rm L}(1;x)$ with $\beta =1$ and $\lambda_1 + 1/2$ a non-negative integer. Two approaches are possible. One, which is to be carried out in \S \ref{S3.2}, is via the recursive approach of \S \ref{S2.2}. For this the difficulty to be overcome is the requirement to involve the analogue of (\ref{7.1b}) for $r+1/2$ a non-negative integer in successive integrations. In \S \ref{S3.3} a second approach will be considered, which involves making use of the Pfaffian structure associated with $f_N^{\rm L}(1;x)$ in the case $\beta =1$ as first identified by Krishnaiah and collaborators \cite{KC71,KW71}. Important here will be the more recent refinement of the associated formulas due to Chiani \cite{Ch14}.

\subsection{Recursive computation of $F_N^{\rm L}(1;x)$ for $\lambda_1 + 1/2$ a non-negative integer}\label{S3.2}
The $r+1/2=:m$ a non-negative integer analogue of (\ref{7.1b}) is
\cite[Eq.~(14)]{Ch14}
\begin{equation}\label{7.1bX}
{1 \over \Gamma(r+1)}
\int_0^x s^r e^{-s} \, ds = 
{\rm erf}(\sqrt{x}) - e^{-x} \sum_{k=0}^{m-1}{x^{1/2 + k} \over \Gamma(k+3/2)}.
\end{equation}
We know from (\ref{5.1c2}) that $F_N^{\rm L}(1;x) = E_N^{\rm L}(0;(x,\infty))$. 
This probability is given by the integration formula (\ref{5.1Y}). For $\beta$ an even integer,  the product $\prod_{1 \le j < k \le N} |x_k - x_j|^\beta$ can be expanded as a multivariate polynomial and thus the multiple integral factorises into univariate integrals. Requiring too that $\lambda_1 + 1/2$ a non-negative integer,
by use of (\ref{7.1bX}) it then follows that in this setting the cumulative distribution $F_N^{\rm L}(1;x)$ will again have a finite structure analogous to (\ref{4a}), but more complicated as the sum over powers of the exponential function will now be a double sum over powers of the exponential function multiplied by powers of ${\rm erf}(\sqrt{\beta x/2})$.

Less clear is the structure in the case $\lambda_1 + 1/2$ a non-negative integer and $\beta$  odd. The product $\prod_{1 \le j < k \le N} |x_k - x_j|^\beta$ can now  be expanded as a multivariate polynomial only after an ordering of the integration variables. This then gives rise to iterations of the integral on the LHS of (\ref{7.1bX}), the structure of which increases in complexity at each iteration. Nonetheless for small values of $\lambda_1 + 1/2$ a non-negative integer and with $\beta =1$, an entirely computer algebra calculation based on the recursive approach of \S \ref{S2.2} can be carried out. This indicates that for $N$ odd $F_N^{\rm L}(1;x)$ can be expanded as
\begin{equation}\label{B1}
F_N^{\rm L}(1;x) =
\sum_{l=1}^{(N+1)/2} \Big ( \sqrt{x} e^{-(2l-1)x/2} p_{l,1}^{\rm o}(x) +
 {\rm erf}(\sqrt{x/2}) e^{-(l-1)x}
p_{l,2}^{\rm o}(x) \Big ),
\end{equation}
for  polynomials $p_{l,1}^{\rm o}(x), p_{l,2}^{\rm o}(x)$ of degree less that $\gamma$ as appears in (\ref{3.2}); specifically $p_{l,2}^{\rm o}(x)|_{l=1}$ is unity.
For $N$ even, one observes the expansion  
\begin{equation}\label{B2}
F_N^{\rm L}(1;x) =
\sum_{l=1}^{N/2+1} \Big (  e^{-(l-1)x} p_{l,1}^{\rm e}(x) +
 \sqrt{x}\, {\rm erf}(\sqrt{x/2}) e^{-(l-1/2)x}
p_{l,2}^{\rm e}(x) \Big ),
\end{equation}
again for polynomials
$p_{l,1}^{\rm e}(x), p_{l,2}^{\rm e}(x)$ of degree less that $\gamma$;
here $p_{l,1}^{\rm e}(x)|_{l=1}$ is
unity.
One highlights that in both (\ref{B1}) and (\ref{B2}) the function ${\rm erf}(\sqrt{x/2})$ does not appear to any higher powers --- at each recursion step, they cancel out. This enables us to utilize the integration formula~\cite{PBM20},
\begin{align}
\nonumber
&\int z^{c} e^{-b^2 z^2}\erf(a z)\,dz=
-\frac{1}{2b^2}z^{c-1}e^{-b^2z^2} \erf(az)\\
&~~~~~ +\frac{a}{b^2\sqrt{\pi}}\int z^{c-1}e^{-(a^2+b^2)z^2}\,dz+\frac{\lambda-1}{2b^2}\int z^{c-2}e^{-b^2 z^2}\erf(az)\,dz,
\end{align}
after adequate variable and parameters substitution.
We record the explicit form of the polynomials for $N=3$ and $\lambda_1 = -1/2$,
\begin{equation}\label{B1a}
p_{1,1}^{\rm o}(x)=\sqrt{2 \over \pi} (1 - x), \quad
p_{2,1}^{\rm o}(x)= - \sqrt{2 \over \pi}, \quad
p_{1,2}^{\rm o}(x)= 1, \quad
p_{2,2}^{\rm o}(x)= - (1 + x).
\end{equation}
and for $N=4$ and $\lambda_1 = -1/2$,
\begin{align}\label{B2a}
& p_{1,1}^{\rm e}(x)= 1, \quad
p_{2,1}^{\rm e}(x)= -2 - {x \over 2} - {x^2 \over 2}- {x^3 \over 4}, \quad p_{3,1}^{\rm e}(x)= 1 + {x \over 2} \nonumber \\
& p_{1,2}^{\rm e}(x)= \sqrt{\pi \over 2} \Big ( - {3 \over 2} + x - {x^2 \over 2} \Big ), \quad
p_{2,2}^{\rm e}(x)= \sqrt{\pi \over 2} \Big (  {3 \over 2} +  {x \over 2} \Big ), \quad p_{2,2}^{\rm e}(x)= 0.
\end{align}

Assuming the validity of expansions (\ref{B1}) and (\ref{B2}), 
from the viewpoint of carrying out the integration in (\ref{3.2}), the required inverse Laplace transforms are standard apart from the need to compute
\begin{equation}\label{B3}
g(p,q;t) := {1 \over 2 \pi i} \int_{c-i \infty}^{c + i \infty}
{{\rm erf}(\sqrt{s}) \over s^{p+1/2}}
e^{s(t - q) } \, ds, \quad c>0
\end{equation}
for $p$ a non-negative integer and $q \ge 0$. Observing by integration by parts that
\begin{align}\label{B3a}
g(p,q;t) & = {(t-q) \over (p - 1/2)}
g(p-1,q;t) + {1 \over 2 \pi i} {1 \over (p-1/2)\sqrt{\pi}} \int_{c-i \infty}^{c + i \infty}
{e^{s(t - 1 - q) } \over s^{p}}
 \, ds \nonumber \\
& = {(t-q) \over (p - 1/2)}
g(p-1,q;t) + { 1 \over (p-1)!\, (p-1/2) \sqrt{\pi}}
(t - 1 - q)^{p-1} \chi_{t > q+1},
\end{align}
we see that the task can be reduced to computing $g(0,q;t)$, for which it can be verified
\begin{equation}\label{B3b}
g(0,q;t) = {\chi_{0 < t - q < 1} \over \sqrt{\pi} \sqrt{t - q}}.
\end{equation}
Therefore, on solving the recurrence relation in \eqref{B3a}, we have
\begin{align}\label{B3b1}
\nonumber
g(p,q;t)&=\sum_{j=1}^p \frac{2^j(2p-2j-1)!!}{\sqrt{\pi}\,(p-j)!\,(2p-1)!!}(t-q-1)^{p-j}(t-q)^{j-1}\chi_{t>q+1}\\
&~~ +\frac{2^p (t-q)^{p-1/2}}{\sqrt{\pi}\,(2p-1)!!}\chi_{0<t-q<1}.
\end{align}
The above is also expressible in terms of the Gauss hypergeometric function, viz.,
\begin{align}
g(p,q,;t)=\frac{(t-q)^{p-1/2}}{\Gamma(p+1/2)}\chi_{t>q}-\frac{(t-q-1)^p}{\sqrt{\pi}\,p!}\,_2F_1(1/2,1;p+1;q-t+1)\,\chi_{t>q+1}. 
\end{align}

As an example, starting with (\ref{B1}) and (\ref{B1a}) we can use the above theory, in particular \eqref{B3b1}, to evaluate the inverse Laplace transform in (\ref{3.2}) and to thus deduce
\begin{multline}\label{B4}
F_3^{\rm fL}(1;x) \Big |_{\lambda_1 = -1/2} = \chi_{1/3<x<1} \Big ( {1 \over 8} \sqrt{x}(35 - 175x + 273 x^2 - 125 x^3)  \\
- (1 - 2x)^{5/2}(1 + 5x) \Theta(x - 1/2) \Big ) + \chi_{x > 1}.
\end{multline}
We can check from (\ref{3.3}) that this is consistent with (\ref{3.4}).

\begin{figure*}[!ht]
  \centering
\includegraphics[width=0.95\textwidth]{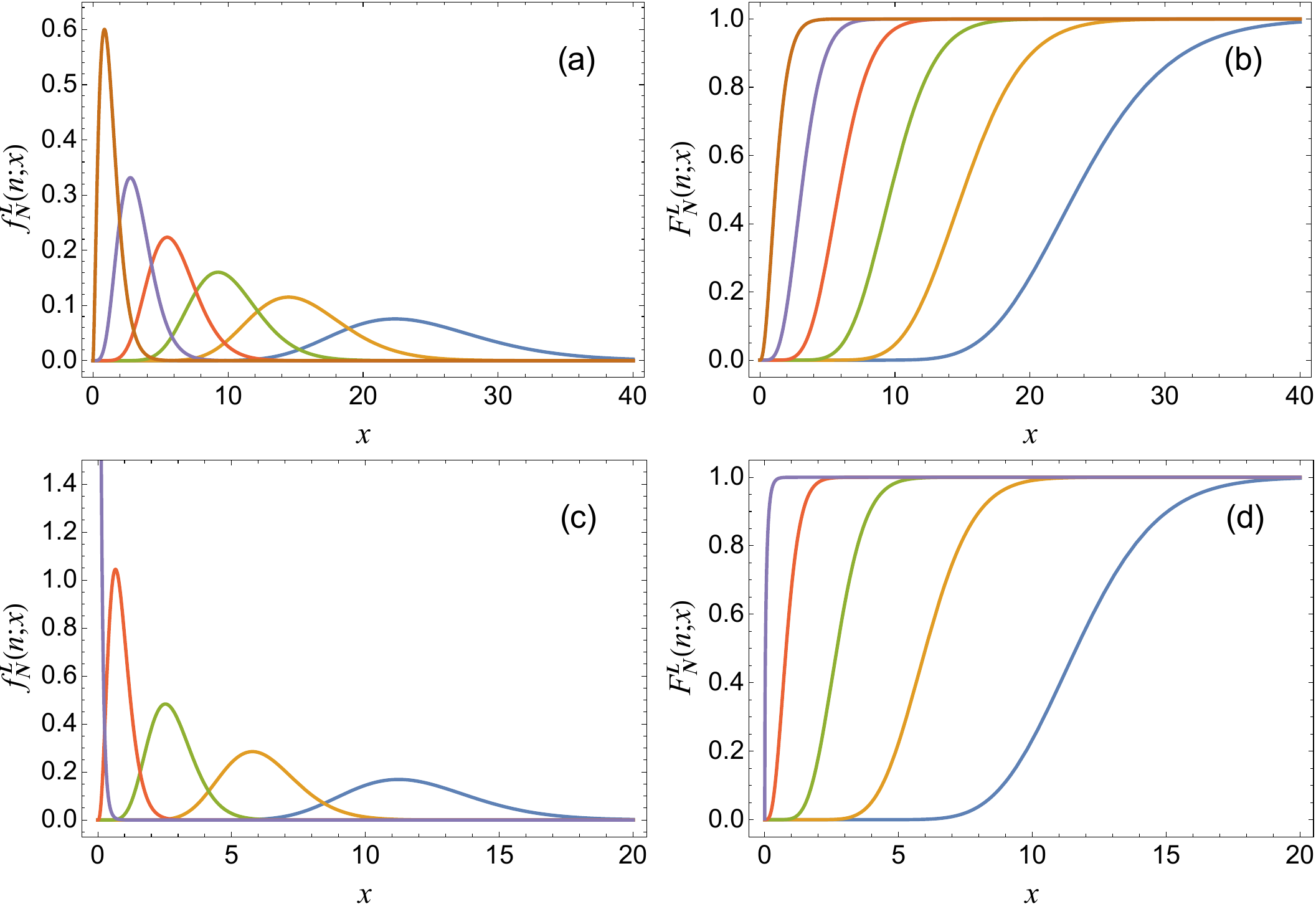}  \caption{Plots of marginal PDFs $f_N^{\rm L}(n;x)$ and marginal CDFs $F_N^{\rm L}(n;x)$ for Laguerre ensemble. Panels (a) and (b) are for parameter choices $N=6,\beta=1,\lambda_1=3/2$, whereas (c) and (d) are for $N=5,\beta=3, \lambda_1=-1/2, \lambda_2=3/4$. In all plots, the curves from right to left correspond to $n=1$ (largest) to $n=N$ (smallest) eigenvalues.}
\label{FigLagBC}
\end{figure*}
\begin{figure*}[!ht]
  \centering
\includegraphics[width=0.95\textwidth]{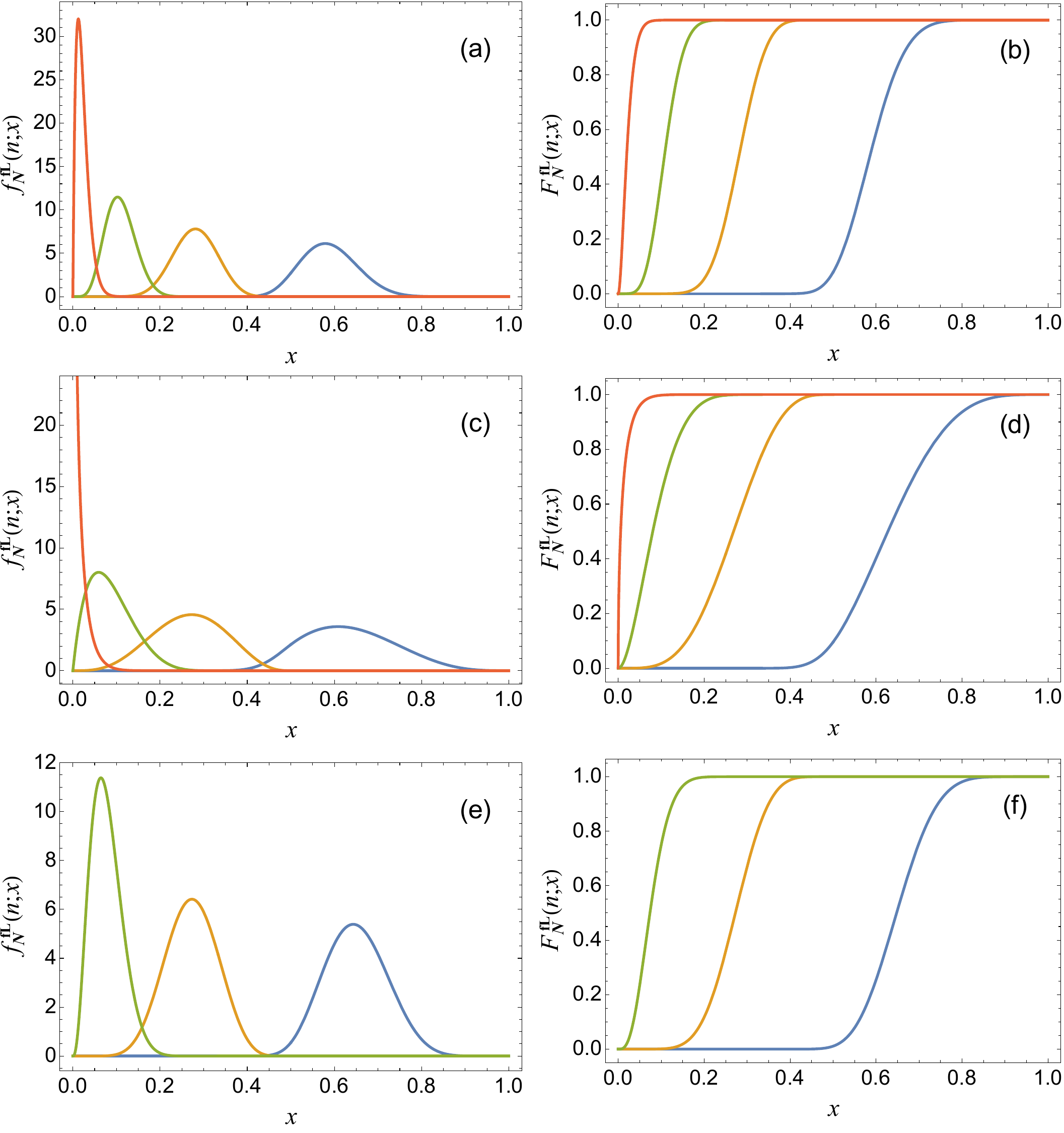}  \caption{Plots of marginal PDFs $f_N^{\rm fL}(n;x)$ and marginal CDFs $F_N^{\rm fL}(n;x)$ for fixed-trace Laguerre ensemble. Panels (a), (b) are for parameter choices $N=4,\beta=3,\lambda_1=1$; (c), (d) are for $N=4,\beta=1, \lambda_1=-1/2$, and (e), (f) are for $N=3, \beta=3, \lambda_1=5/2$. In all plots, the curves from right to left correspond to $n=1$ (largest) to $n=N$ (smallest) eigenvalues.}
\label{FigFTLagABC}
\end{figure*}

Interestingly, computational symbolic algebra reveals that even for $\beta >1$ odd positive integers, the eigenvalue distributions exhibit structures akin to those described by equations (3.8) and (3.9), involving solely the first power of $\erf(\sqrt{\beta x/2})$. Leveraging this insight, we can once again employ the inverse-Laplace transform relations~\eqref{3.2} and~\eqref{B3b1} to derive the distributions for the corresponding fixed-trace variant Laguerre ensemble. Moreover, we can now obtain the exact conductance distribution for scenarios where $|N_1-N_2|$ is an even integer -- a previously intractable case (with $\lambda_1 = (n_1-N+1)/2-1$, $n_1 = \max \{N_1,N_2\}$ and $N = \min \{N_1,N_2\}$; the quantities $N_1$ and $N_2$ are the number of channels in the leads of the quantum dot).

\begin{figure*}[!ht]
  \centering
\includegraphics[width=0.98\textwidth]{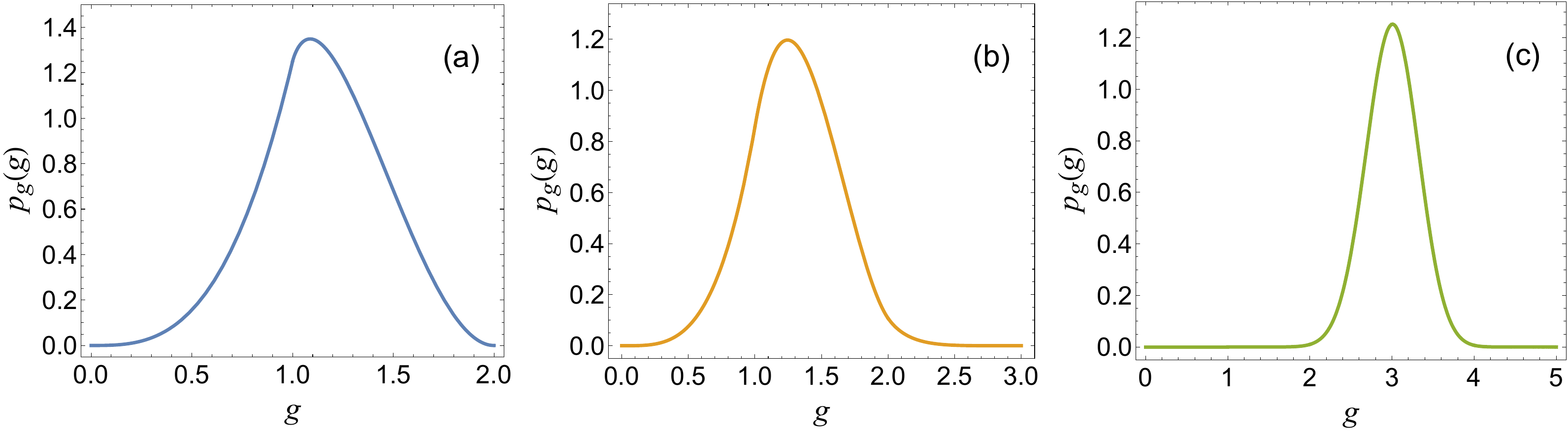}  \caption{Plots of Landauer conductance PDF for number of channels (a) $N_1=2, N_2=4$, (b) $N_1=3, N_2=3$, and (c) $N_1=5, N_2=9$.}
\label{FigCondDist}
\end{figure*}

In the supplementary material, we provide Mathematica codes ``{\sf Laguerre-B}" and ``{\sf Laguerre-C}" which compute the marginal eignvalue distributions for the case $\lambda_1+1/2$ a non-negative integer, when  $\beta=1$ and $\beta$ a positive odd-integer (including 1), respectively. The corresponding codes for fixed trace variants are labelled as ``{\sf Laguerre-FixedTrace-A}" (associated with ``{\sf Laguerre-A}" mentioned before), ``{\sf Laguerre-FixedTrace-B}", and ``{\sf Laguerre-FixedTrace-C}". The Mathematica code for obtaining the above mentioned conductance distribution is contained in the file ``{\sf CondDist}". Figures~\ref{FigLagBC} and \ref{FigFTLagABC} showcase example plots for the above mentioned unconstrained and fixed-trace Laguerre cases, respectively. Moreover, Fig. 5 presents example plots for the PDF of Landauer conductance.

\subsection{Pfaffian form of $F_N^{\rm L}(1;x)$}\label{S3.3}
As noted in the Introduction, the Pfaffian structure of $F_N^{\rm L}(1;x)$ was first identified in the work of Krishnaiah and collaborators \cite{KC71,KW71}. Our starting point is a subsequent refinement of this early work due to Chiani \cite{Ch14}.

\begin{prop}\label{P1}
In the context of the real Wishart matrix $W_{n_1,N}^{(1)}$, set $\alpha=: (n_1-N-1)/2$. 
For $l$ a non-negative integer define
\begin{align}\label{P6a}
P(1/2 + l;x) & := {\rm erf}(\sqrt{x}) - e^{-x} \sum_{k=0}^{l-1} {x^{1/2+k} \over \Gamma(1/2 + k + 1)}, \\\label{P6b}
P(l;x) & := 1 - e^{-x} \sum_{k=0}^{l-1} {x^{k} \over \Gamma( k + 1)}.
\end{align}
Use this in the further definitions
\begin{align}\label{P7a}
& p_l:=P(\alpha+l;x/2), \quad \gamma_l:=\Gamma(\alpha+l), \nonumber
\\
& q_l := 2^{-(2 \alpha + l)} \Gamma(2 \alpha + l) P(2\alpha+l;x), \quad
r_l := e^{-x/2} (x/2)^{\alpha+l} /\Gamma(\alpha+1+l),
\end{align}
and from these quantities set
\begin{equation}\label{P7}
a_{i,i+k} =  - p_i \sum_{l=i}^{i+k-1} r_l +
{2 \over \gamma_i} \sum_{l=i}^{i+k-1} {q_{i+l} \over \gamma_{l+1}}, \quad i=1,\dots,N, \: k=1,\dots,N- i.
\end{equation}
For $N$ odd define too
\begin{equation}\label{P4}
a_{i,N+1}(x) = {2^{-\alpha - N - 1} \over \Gamma(\alpha + N + 1)}p_i, \quad i=1,\dots,N.
\end{equation}
Further set $a_{i,i} = 0$, $a_{j,i} = - a_{i,j}$ for $j \ge i$ so in particular the matrix $[a_{i,j}]$ is anti-symmetric.

Up to a known normalisation (which we omit for brevity --- this is uniquely determined by the requirement that $F_N^{\rm L}(1;x) \to 1$ as $x \to \infty$), one has that for $N$ even
\begin{equation}\label{P5}
F_N^{\rm L}(1;x) \propto 
{\rm Pf}  [ a_{i,j}(x) ]_{i,j=1,\dots,N},
\end{equation}
while for $N$ odd 
\begin{equation}\label{P6}
F_N^{\rm L}(1;x) \propto 
{\rm Pf} \left [
\begin{array}{cc}
[ a_{i,j} ]_{i,j=1,\dots,N} &
[a_{N+1,j}]_{j=1,\dots,N} \\{}
 [a_{i,N+1}]_{i=1,\dots,N} & 0
\end{array} \right ].
\end{equation}
\end{prop}

From the definition of a Pfaffian as a sum of products of its elements (see e.g.~\cite[Definition 6.1.4]{Fo10}) it would appear from Proposition \ref{P1} that the expanded form of $F_N^{\rm L}(1;x)$ is more complicated than given in the conjectures (\ref{B1}) and (\ref{B2}), involving in particular $({\rm erf}(\sqrt{x/2}))^r$ for $r=1,\dots,[(N+1)/2]$. But this does not take into account possible cancellations, which indeed take place as is seen by explicitly evaluating (\ref{P5}) and (\ref{P6}) for small $N$, giving consistency with (\ref{B1}) and (\ref{B2}). How to anticipate this though still remains to be understood.

As noted in  \cite{Ch14}, the Pfaffian forms (\ref{P5}) and (\ref{P6}), with their elements permitting a recursive computation, are well suited to tabulating $F_N^{\rm L}(1;x)$ as a function of $x$ for given $N$ up to moderately large values.

\subsection{Functional form of $\{F_N^{\rm L}(n;x) \}$, $n \ge 2$ }
The functional form of $F_N^{\rm L}(n;x)$ for $n$ even, $N$ odd, $\beta =1$ and $\lambda_1+1/2$ a non-negative integer can be explicitly determined. For this purpose, order the eigenvalues $\lambda_1 > \lambda_2 > \cdots > \lambda_N$. Denote the eigenvalue distribution of the Laguerre ensemble with $\beta =1$ and parameter $\lambda_1$ by ${\rm OE}_N(x^{\lambda_1} e^{-x/2})$ --- thus $x^{\lambda_1} e^{-x/2}$ is the weight function, $N$ is the total number of eigenvalues and ``OE'' denotes orthogonal ensemble, with real Wishart matrices being invariant under transformation by real orthogonal matrices. Denote too the eigenvalue distribution of the Laguerre ensemble with $\beta =4$ and weight $x^\lambda e^{-\lambda_1}$ by ${\rm SE}_N(x^{\lambda_1} e^{-x})$. Here ``SE'' denotes symplectic ensemble, with quaternion Wishart matrices being invariant under transformation by symplectic unitary matrices.

\begin{prop}
We have
\begin{equation}\label{3.23}
F_{2N+1}^{\rm L}(2n;x)\Big |_{\lambda_1 = (a-1)/2 \atop \beta = 1} =
F_N^{\rm L}(n;x)\Big |_{\lambda_{1}= a + 1 \atop \beta = 4}.
\end{equation}
Consequently for $\lambda_1 + 1/2$ a non-negative integer $F_{2N+1}^{\rm L}(2n;x)|_{\beta = 1}$ exhibits the functional form (\ref{4a}).
\end{prop}

\begin{proof}
We know from \cite{FR01} that
$$
{\rm even} \, {\rm OE}_{2N+1}(x^{(a-1)/2} e^{-x/2}) =
{\rm SE}_N(x^{a+1} e^{-x}),
$$
where the notation ``even'' refers to the distribution of all the even labelled eigenvalues. The identity (\ref{3.23}) between marginal eigenvalue distributions is an immediate consequence of this joint relation. Suppose now $\lambda_1  = (a-1)/2$ and
$\lambda_1 + 1/2$ a non-negative integer. Then we have that $a/2$ is a non-negative integer and so $a+1$ is an odd positive integer. But for $a+1$ any non-negative integer,
the functional form of the RHS of (\ref{3.23}) is given by (\ref{4a}), thus implying the stated functional  form of 
$F_{2N+1}^{\rm L}(2n;x) |_{\beta = 1}$.
\end{proof}

\begin{remark}
    Setting $N=n=1$ in (\ref{3.23}) shows
    \begin{equation}
    F_3^{\rm L}(2;x) \Big |_{\lambda_1 = (a-1)/2} = 
    {1 \over \Gamma(a+2)} \int_0^x t^{a+1}
    e^{-t} \, dt.
    \end{equation}
    As noted in \cite{Da72a}, this simple functional form  was first obtained by Eckert \cite{Ek74}.
\end{remark}

Thus for an odd number of total eigenvalues, $F_{2N+1}^{\rm L}(2n;x) |_{\beta = 1}$ in the case that $\lambda_1 + 1/2$ a non-negative integer does not involve the error function, unlike the functional form of $F_{2N+1}^{\rm L}(1;x) |_{\beta = 1}$. On the other hand computation of
the cumulative distribution of the second largest eigenvalue with an even total number of eigenvalues, i.e.~$F_{2N}^{\rm L}(2;x)$, in the case $\lambda_1 = -1/2$ for definiteness using the recurrence strategy of \S \ref{S3.2}, one observes again the functional form (\ref{B2}), but with the upper terminal reduced by one.

A special functional form of $f_{N}^{\rm L}(N;x) |_{\beta = 1}$ --- which is the probability density function of the smallest eigenvalue --- in the case $\lambda_1 + 1/2$ a non-negative integer is known from the work of Edelman \cite[Th.~2.1]{Ed91}. To state this, let $U(a,b;z)$ denote the confluent hypergeometric function of the second kind (Tricomi function). One then has
\begin{equation}\label{3.24}
f_{N}^{\rm L}(N;x) = 
x^{\lambda_1} e^{-N x/2}
\bigg ( P(x) U \Big ( {N -1 \over 2},
- {1 \over 2}; {x \over 2} \Big ) +
Q(x) U \Big ( {N + 1 \over 2},
 {1 \over 2}; {x \over 2} \Big ) \bigg ),
 \end{equation}
 for certain polynomials $P(x), Q(x)$, which for $\lambda_1 = -1/2$ are $Q(x)=0$ and $P(x)$ a positive constant. The latter form initial conditions from which these polynomials can be computed by recurrence for general $\lambda_1 + 1/2$ a non-negative integer. It is shown in \cite{Ed91} that this form permits hypergeometric expressions for the moments of $f_{N}^{\rm L}(N;x)$. However, it  does not yield an explicit evaluation of the corresponding fixed trace probability density, with there being no obvious way to compute the inverse Laplace transform as required by (\ref{3.2}). In fact it is possible to expand (\ref{3.24}) in a form analogous to (\ref{B1}) and (\ref{B2}), which is simpler in that it involves only two distinct exponential functions, one of which has an error function factor. Such an expansion can be deduced from the three term recurrence \cite[Eq.~13.3.7]{DLMF}
 $$
 -a(a-b+1)U(a+1,b;z) = U(a-1,b;z) + (b - 2a - z)U(a,b;z), 
 $$
supplemented by initial conditions which can be generated by computer algebra evaluation, for example
$$
U(-1/2,1/2;z) = \sqrt{z}, \quad 
U(1/2,1/2;z) = \sqrt{\pi} e^z {\rm erfc}(\sqrt{z}).
$$

\begin{figure*}[!ht]
  \centering
\includegraphics[width=0.98\textwidth]{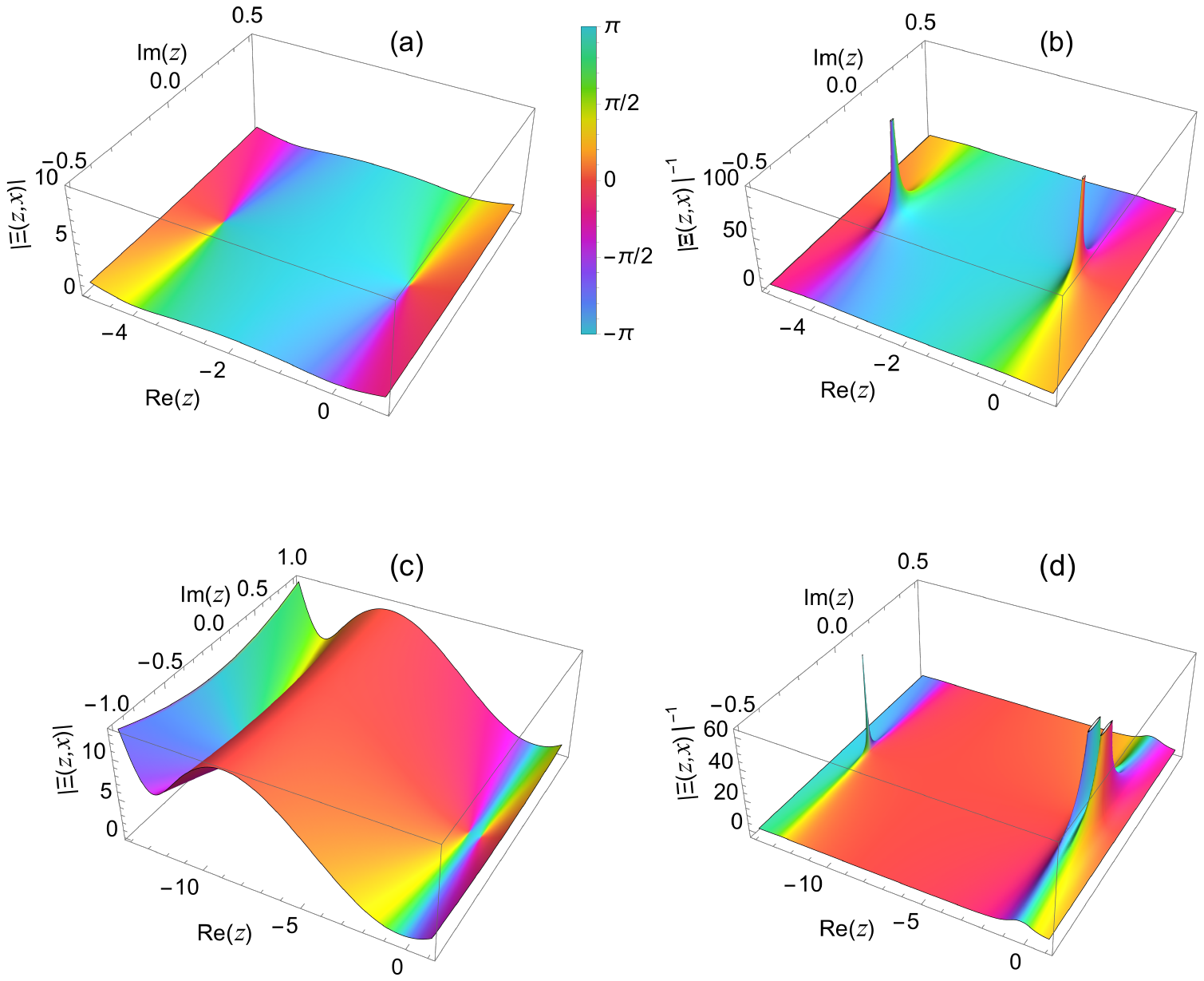}  \caption{Plots of the modulus of generating function $\Xi_N(z;x)$ and that of its inverse $1/\Xi_N(z;x)$ with respect to $z$ in the complex plane. The corresponding arguments are depicted using colors, as indicated in the legend, applicable to all plots. Panels (a), (b) show these quantities for the Jacobi case with $N=2, \beta=4,\lambda_1=1/2$, $\lambda_2=1$, and choice of $x=1/4$. Similarly, panels (c), (d) show the results for the Laguerre case with $N=3,\beta=1,\lambda=1$, $x=4$. In the former case, the zeros of $\Xi_N(z;x)$ and hence the simple poles of $1/\Xi_N(z;x)$ occur at $z\approx-4.026, -5.515\times 10^{-4}$. In the latter case, the zeros (simple poles) of $\Xi_N(z;x)$ (resp. $1/\Xi_N(z;x))$  occur at $z\approx-12.71, -6.237\times 10^{-1},-2.577\times 10^{-2}$.}
\label{FigGFZ}
\end{figure*}

\subsection{Zeros of the generating function $\Xi_N(z;x)$}

A well known consequence of the $\beta =2$ Fredholm determinant expression (\ref{E2}) is that all the zeros of the generating function (\ref{E1}) as a function of $z$ lie on the negative real axis; see e.g.~\cite{HKPV09}. In \cite{FL14} this analytic property was used to establish that for $N \to \infty$ and with certain so-called soft edge scaling of $x$, the probabilities in $\{E_N(n;(x,b)) \}$, satisfy a local central limit theorem. The essential mechanism is that the zeros being on the negative real axis imply that the probabilities are log concave,
$$
\log E_N(n+1;(x,b)) -2 \log E_N(n;(x,b)) + \log E_N(n-1;(x,b)) \le 0,
$$
which in turn is a sufficient condition for the passage from a central to a local limit theorem \cite{Be73}.

It is well known that for $\beta =1$ and $\beta =4$, due to an underlying Pfaffian point process structure, an analogue of (\ref{E2}) holds true,
\begin{equation}
\Big ( \Xi(1-z;x) \Big )^2 =
\det (\mathbb I - z J^{-1} A ), \qquad
J = \begin{bmatrix} 0 & 1 \\ -1 & 0 \end{bmatrix},
\end{equation}
where $A$ is a real $2 \times 2$ antisymmetric integral operator (this equation is equivalent to \cite[Eq.~(6.27)]{Fo10}). However, the matrix integral operator $J^{-1} A$ is not self adjoint, so we cannot conclude from this the  location of the zeros from an eigenvalue factorisation. We remark that in the work \cite{Ka14} a Pfaffian point process for which the eigenvalues of $J^{-1} A$ are all real was said to have a kernel with a quasi-real diagonal form, and moreover, this property was used to guarantee the existence of the Pfaffian point process as specified by the integral operator $A$. However, in this same work no progress was made on establishing the quasi-real diagonal form for integral operators $A$ arising in random matrix theory.

We can made use of our codes to compute the zeros of $ \Xi(z;x)$ from cases that we have generated the sequence of cumulative distributions $\{ F_N(n;x) \}$. Thus we have from (\ref{5.1c2}) and (\ref{E1}) that
\begin{equation}\label{E4}
\Xi_N(z;x) := (1-z) \sum_{n=1}^N z^{n-1} F_N(n;x) + z^N.
\end{equation}
To proceed further, we fix $x$ and numerically solve the polynomial equation in $z$, {viz.}, $\Xi_N(z;x) = 0$. Due to possible severe ill-conditioning (depending on $x$ the roots may be huge in magnitude, even though $N$ is small, or they may be tiny in magnitude --- examples are given in the Mathematica file ``{\sf GeneratingFunctionZeros}" contained in the supplementary material)
the numerical values of $\{F_N(n;x) \}$ must be given to high precision. With this done, we have found that in all cases we have investigated (Laguerre, Jacobi with $\lambda_1$ a non-negative integer up to 10 and $\beta$ a positive integer up to 8 and $N$ taking values up to 12, as well as the Laguerre case with $\lambda_1+1/2$ a non-negative integer 0 or 1 and $\beta =1$, and its fixed trace version)
that all the zeros of $\Xi_N(z;x)$ are on the negative real axis. 
There is a qualification to this statement in the fixed trace Laguerre case. For $x < 1/N$, due to the constraint $\sum _{l=1}^N x_l = 1$, the gap probability $E_N (0; (x, 1))$ must vanish, so we see from \eqref{E1} that $z = 0$ is also a zero. As soon as $x>1/N$, $E_N (0; (x, 1))$ becomes non-vanishing, however now $E_N (N; (x, 1))$ vanishes thereby reducing the degree of the polynomial $\Xi_N(z;x)$ in $z$ by 1. Similarly, as $x$ gets closer to $1$, gradually all $E_N (n; (x, 1))$ except $n=0$ and $n=1$ vanish and we are left with only one zero of $\Xi_N(z;x)$, given by $-E_N (0; (x, 1))/E_N (1; (x, 1))$. We also note that for $x\to 1$, we have $E_N (0; (x, 1))\to 1$ and $E_N (1; (x, 1))\to 0$, so that this zero tends to $-\infty$.
In Fig.~\ref{FigGFZ}, we show two examples where the zeros of $\Xi_N(z;x)$ (equivalently, simple poles of $1/\Xi_N(z;x)$) can be visualized.

\subsection*{Acknowledgements}
The work of PJF was supported by the Australian Research Council  grant DP210102887. SK acknowledges the support provided by SERB, DST, Government of India, via Grant No. CRG/2022/001751.

\providecommand{\bysame}{\leavevmode\hbox to3em{\hrulefill}\thinspace}
\providecommand{\MR}{\relax\ifhmode\unskip\space\fi MR }
\providecommand{\MRhref}[2]{%
  \href{http://www.ams.org/mathscinet-getitem?mr=#1}{#2}
}
\providecommand{\href}[2]{#2}


\begin{thebibliography}{10}

\bibitem{AKT19}
S.~Adachi, H.~Kubotani and M.~Toda,
\emph{Exact distribution of largest Schmidt eigenvalue for quantum entanglement},  J. Phys. A \textbf{52} (2019), 405304.

    \bibitem{BF97a}
T.H. Baker and P.J. Forrester, \emph{The {Calogero-Sutherland} model and
  generalized classical polynomials}, Commun. Math. Phys. \textbf{188} (1997),
  175--216.

\bibitem{Be97}
C.W.J. Beenakker, \emph{Random-matrix theory of quantum transport}, Rev. Mod.
  Phys. \textbf{69} (1997), 731--808.

  \bibitem{Be73}
E.A. Bender, \emph{Central and local limit theorems applied to asymptotic
  enumeration}, J. Combin. Theory Ser. A \textbf{15} (1973), 91--111.

    \bibitem{Bo09}
F.~Bornemann, \emph{On the numerical evaluation of distributions in random
  matrix theory: a review},
  Markov Processes Relat. Fields \textbf{16} (2010), 803--866.

\bibitem{Ch14}
M.~Chiani,
\emph{Distribution of the largest eigenvalue for real Wishart and
Gaussian random matrices and a simple approximation for
the Tracy-Widom distribution},
J. Multivariate Anal. \textbf{129} (2014), 69--81.

     \bibitem{Da68} 
 A.W.~Davis, \emph{A system of linear differential equations for the distribution of
 Hotelling's generalized $T_0^2$}, Ann. Math. Statistics \textbf{39} (1968), 815--832.   

   \bibitem{Da72a}
A.W. Davis, \emph{On the marginal distributions of the latent roots of the
  multivariable beta matrix}, Ann. Math. Statist. \textbf{43} (1972),
  1664--1669.

\bibitem{Da72b}
A.W. Davis, \emph{
  On the distributions of the latent roots and traces of certain random matrices}, J. Multivariate Anal. \textbf{2} (1972), 189--200.

  \bibitem{DMJ03} P. A. Dighe, R. K. Mallik, and S. S. Jamuar, 
  \emph{Analysis of transmit-receive diversity in Rayleigh
fading},
IEEE Trans. Commun. {\bf 51}, 694 (2003).

\bibitem{DE02}
I.~Dumitriu and A.~Edelman, \emph{Matrix models for beta ensembles}, J. Math.
  Phys. \textbf{43} (2002), 5830--5847.

  \bibitem{Ek74}
  S.R.~Eckert, \emph{Distributions of the individual ordered roots of random matrices}, PhD. thesis,
  University of Adelaide, 1974.

    \bibitem{Ed91}
  A.~Edelman, \emph{The distribution and moments of the smallest eigenvalue of
  a random matrix of Wishart type}, Lin.~Alg. Appl. {\bf 159} (1991), 55--80.  

   \bibitem{ES08}
A.~Edelman and B.D. Sutton, \emph{The beta-{J}acobi matrix model, the {CS}
  decomposition, and generalized singular value problems}, Found. Comput. Math.
  \textbf{8} (2008), 259--285.

  \bibitem{Fo93}
P.J. Forrester, \emph{Recurrence equations for the computation of correlations in the
  $1/r^2$ quantum many body system}, J. Stat. Phys. \textbf{72} (1993), 39--50.

    \bibitem{Fo93c}
P.J. Forrester, \emph{Exact results and universal asymptotics in the {Laguerre} random
  matrix ensemble}, J. Math. Phys. \textbf{35} (1993), 2539--2551.


  \bibitem{Fo10}
P.J. Forrester, \emph{Log-gases and random matrices}, Princeton University
  Press, Princeton, NJ, 2010.

    \bibitem{FI10a}
P.J. Forrester and M.~Ito, \emph{Difference system for {S}elberg correlation
  integrals}, J. Phys. A \textbf{43} (2010), 175202.

  \bibitem{FK19}
P.J. Forrester and S.~Kumar, 
\emph{Recursion scheme for the largest $\beta$-Wishart-Laguerre
eigenvalue and Landauer conductance in quantum transport}, J. Phys. A \textbf{52} (2019), 42LT02. 

\bibitem{FK23}
P.J. Forrester and S.~Kumar, 
\emph{Computable structural formulas for the distribution of the
$\beta$-Jacobi edge eigenvalues},
The Ramanujan J. \textbf{61} (2023), 87--110.

  \bibitem{FL14} 
  P. Forrester, J. Lebowitz,  \emph{Local central limit theorem for determinantal point
processes}, J. Stat. Phys. \textbf{157}  (2014), 60--69.

  \bibitem{FR01}
P.J. Forrester and E.M. Rains, \emph{Inter-relationships between orthogonal, unitary and symplectic matrix
ensembles}, In P.M. Bleher and A.R. Its, editors,
\emph{Random matrix models and their applications}. volume 40 of \emph{Mathematical Sciences 
Research Institute Publications}, pages 171-208. Cambridge University Press, United Kingdom, 2001.

  \bibitem{FR02b}
P.J. Forrester and E.M. Rains, \emph{Interpretations of some parameter
  dependent generalizations of classical matrix ensembles}, Prob. Theory
  Related Fields \textbf{131} (2005), 1--61.

  \bibitem{FR12} P.J. Forrester and E.M. Rains, 
\emph{A Fuchsian matrix differential equation for Selberg correlation integrals},
Commun. Math. Phys. \textbf{309}, 771 (2012).


\bibitem{FT19} P.J. Forrester and A.K. Trinh,
\emph{Finite size corrections at the hard edge for the Laguerre $\beta$ ensemble},
Stud. Appl. Math. \textbf{143} (2019), 315--336.

\bibitem{FW01a}
P.J. Forrester and N.S. Witte, \emph{Application of the $\tau$-function theory of {Painlev\'e}
  equations to random matrices: {PV}, {PIII}, the {LUE}, {JUE} and {CUE}},
  Commun. Pure Appl. Math. \textbf{55} (2002), 679--727.

  \bibitem{FW04}
P.J. Forrester and N.S. Witte, \emph{Application of the $\tau$-function theory of {Painlev\'e}
  equations to random matrices: {PVI}, the {JUE},{CyUE}, {cJUE} and scaled
  limits}, Nagoya Math. J. \textbf{174} (2004), 29--114.


\bibitem{GN99}
A.K. Gupta and D.K. Nagar, \emph{Matrix variate distributions}, Chapman \&
  Hall/CRC, Boca Raton, FL, 1999.

     \bibitem{HKPV09}
J.B. Hough, M.~Krishnapur, Y.~Peres, and B.~Vir\'ag, \emph{Zeros of {G}aussian
  analytic functions and determinantal point processes}, American Mathematical
  Society, Providence, RI, 2009.

 \bibitem{Ka14}
 V.~Kargin, \emph{On Pfaffian random point fields}, J. Stat.
Phys. \textbf{154} (2014), 681--704.  
  

\bibitem{KC71}
P.R.~Krishnaiah and T.C.~Chang, \emph{On the exact distributions of the extreme
roots of the Wishart and MANOVA matrices}, J. Multivariate Analysis,
\textbf{1} (1971), 108--117.

\bibitem{KW71}
P.R.~Krishnaiah and V.B.~Waiker, \emph{Exact joint distributions of any few ordered roots of a class of random
matrices}, J. Multivariate Anal. \textbf{1} (1971), 308--315.

\bibitem{Ku19} S. Kumar, \emph{Recursion for the Smallest Eigenvalue Density of beta-Wishart-Laguerre Ensemble},
  J. Stat. Phys. {\bf 175}, (2019) 126.

  \bibitem{KP2010} S. Kumar and A. Pandey, \emph{Random Matrix Model for Nakagami–Hoyt Fading}, IEEE Trans. Inf. Theory \textbf{56} (2010), 2360--2372.

   \bibitem{KP11a}
  S.~Kumar and A.~Pandey, \emph{Entanglement in random pure states: spectral density
  and average von Neumann entropy}, J.~Phys. A  \textbf{44} (2011), 445301.

  


\bibitem{LP88}
S.~Lloyd and H.~Pagels, \emph{Complexity as thermodynamic depth},
Annals Phys.  {\bf 188}, (1988) 186--213.



\bibitem{Mtmk} {\it Mathematica} (Wolfram Research, Inc., Champaign, 2022), Version 13.1.



   \bibitem{Me67}
M.L. Mehta, \emph{Random matrices and the statistical theory of energy levels},
  Academic Press, New York, 1967.

  \bibitem{Mu82}
R.J. Muirhead, \emph{Aspects of multivariate statistical theory}, Wiley, New
  York, 1982.

    \bibitem{DLMF}  NIST Digital Library of Mathematical Functions. https://dlmf.nist.gov

\bibitem{Pi76}
  K.C.S.~Pillai, \emph{Distributions of characteristic roots in multivariate analysis part I. Null distributions},
Canadian J. Statist./ La Revue Canadienne de Statistique, \textbf{4} (1976), 
157--183.

\bibitem{PD69}
K.C.S.~Pillai and C.O.~Dotson, 
\emph{Power comparisons of tests of two multivariate hypotheses based on individual characteristic roots},
Ann. Inst. Statist. Math. \textbf{21} (1969), 49--66.



  \bibitem{PBM20}
A.P. Prudnikov, Yu. A. Brychkov and O.I. Marichev, \emph{Integrals and Series, Vol. 2, Special Functions}, Gordon and Breach Science Publishers, NY, 1992.



\bibitem{Ro45}
S.N. Roy, \emph{The individual sampling distribution of the maximum, the minimum,
and any intermediate of the $p$-statistics on the null-hypothesis} Sankhy\={a} \textbf{7} (1945),
133--158.


\bibitem{SBL22}
B. Sharmila, V. Balakrishnan, and S. Lakshmibala,
\emph{Exact eigenvalue order statistics for the reduced density matrix of a bipartite system},
Ann. Phys.
\textbf{446},
(2022), 
169107.

\bibitem{TW94c}
C.A. Tracy and H.~Widom, \emph{Fredholm determinants, differential equations and matrix
  models}, Commun. Math. Phys. \textbf{163} (1994), 33--72.

  \bibitem{TV04}
A.M. Tulino and S.~Verd\'u, \emph{Random matrix theory and wireless
  communications}, Foundations and {T}rends in {C}ommuncations and
  {I}nformation {T}heory, vol.~1, Now Publisher, (2004), pp.~1--182.

   \bibitem{Vi11}
P.~Vivo, \emph{Largest Schmidt eigenvalue of random pure states and conductance distribution in chaotic cavities},
J. Stat. Mech. \textbf{2011} 2011, P01022.

\bibitem{ZCW09}
A.~Zanella, M. Chiani, and M.Z. Win. \emph{On the marginal distribution of the eigenvalues
of Wishart matrices}, IEEE Transactions on Communications \textbf{57} (2009), 1050--1060. 


\bibitem{ZS2001} K. \.Zyczkowski and H.-J. Sommers, \emph{Induced measures in the space of mixed quantum states}, J. Phys. A: Math. Gen. \textbf{34} (2001), 7111.





\end{thebibliography}
\end{document}